\documentclass[acmsmall]{acmart}
\usepackage{amsmath,amsfonts}
\usepackage{graphicx}
\usepackage{algorithm}
\usepackage{algpseudocode}
\usepackage{subcaption}
\usepackage{textcomp}
\usepackage{multirow}
\usepackage{dafny}
\usepackage{xcolor}
\usepackage{tikz}
\usepackage{adjustbox}
\usepackage{array}
\usepackage{comment}
\usepackage{parcolumns}
\usepackage{wrapfig}
\usepackage{hyperref}
\usepackage{bbding}
\usepackage{float}
\usepackage{mdframed}
\usepackage{xcolor}
\usepackage[T1]{fontenc}
\usetikzlibrary{tikzmark,calc}
\usepackage{tcolorbox}
\tcbuselibrary{skins, breakable}
\usepackage{enumitem}
\usepackage{array}
\usepackage{pifont} 
\usepackage{mathpartir}
\usetikzlibrary{arrows.meta, positioning, matrix}  

\theoremstyle{definition}

\newtheorem{lemma}{Lemma}
\newtheorem{theorem}{Theorem}

\newif\ifshowcomments

\definecolor{programs}{gray}{0.1}
\definecolor{keywords}{HTML}{204a87}
\definecolor{comments}{HTML}{8f5902}
\definecolor{strings}{HTML}{4e9a06}
\lstdefinestyle{standard}{
	escapechar=@,
	basicstyle=\ttfamily\small\color{programs},
	commentstyle = \ttfamily\color{comments},
	keywordstyle=\ttfamily\color{keywords},
	stringstyle=\color{strings},
	xleftmargin=5.0ex,
	columns=fullflexible,
	showstringspaces=false,
	showspaces=false,
    keepspaces=true,
    alsoletter={()[].=:}
}
\newcommand\code{\lstinline[basicstyle=\ttfamily\small,breakatwhitespace,xleftmargin=0pt]}

\AtBeginDocument{%
  }

\setcopyright{cc}
\setcctype{by}
\acmDOI{10.1145/3776705}
\acmYear{2026}
\acmJournal{PACMPL}
\acmVolume{10}
\acmNumber{POPL}
\acmArticle{63}
\acmMonth{1}
\received{2025-07-10}
\received[accepted]{2025-11-06}

\begin{document}

\title{DafnyMPI: A Dafny Library for Verifying Message-Passing Concurrent Programs}

\newcommand{\tuftsa}{\affiliation{
  \institution{Tufts University}
  \city{Medford}
  \state{Massachusetts}
  \country{USA}
}}

\newcommand{\auca}{\affiliation{
  \institution{American University of Central Asia}
  \city{Bishkek}
  \country{Kyrgyzstan}
}}

\author{Aleksandr Fedchin}
\email{aleksandr.fedchin@tufts.edu}
\orcid{0000-0003-0810-1941}
\tuftsa
\auca

\author{Antero Mejr}
\email{antero.mejr@tufts.edu}
\orcid{0009-0000-8124-467X}
\tuftsa

\author{Hari Sundar}
\email{hari.sundar@tufts.edu}
\orcid{0000-0001-9001-5107}
\tuftsa

\author{Jeffrey S. Foster}
\email{jfoster@cs.tufts.edu}
\orcid{0000-0001-8043-1166}
\tuftsa

\begin{abstract}
The Message Passing Interface (MPI) is widely used in parallel, high-performance
programming, yet writing bug-free software that uses MPI remains difficult.
We introduce DafnyMPI, a novel, scalable approach to formally verifying MPI software.
DafnyMPI allows proving deadlock freedom, termination, and functional equivalence with simpler sequential implementations.
In contrast to existing specialized frameworks, DafnyMPI avoids custom concurrency logics and instead relies on Dafny, a verification-ready programming language used for sequential programs, extending it with concurrent reasoning abilities.
DafnyMPI is implemented as a library that enables safe MPI programming by requiring users to specify the communication topology upfront and to verify that calls to communication primitives such as \texttt{MPI\_ISEND} and \texttt{MPI\_WAIT} meet their preconditions.
We formalize DafnyMPI using a core calculus and prove that the preconditions suffice to guarantee deadlock freedom.
Functional equivalence is proved via rely-guarantee reasoning over message payloads and a system that guarantees safe use of read and write buffers. 
Termination and the absence of runtime errors are proved using standard Dafny techniques.
To further demonstrate the applicability of DafnyMPI, we verify numerical solutions to three canonical partial differential equations.
We believe DafnyMPI demonstrates how to make formal verification viable for a broader class of programs and provides proof engineers with additional tools for software verification of parallel and concurrent systems.
\end{abstract}




\maketitle

\section{Introduction}

Computational science relies heavily on parallelism to achieve high performance and enable large-scale simulation, forecasting, and decision-making.
The correctness of parallel scientific software is critical, yet there are few practical formal verification techniques for such systems.
Model checking and dynamic verification approaches~\cite{isp,verifympi2,verifympi3,verifympifinal,civl} can suffer from the state explosion problem 
and typically require assuming a predetermined input and number of processes. 
Proof assistants, e.g., those using concurrent separation logic~\cite{iris,actris}, can perform arbitrarily complex reasoning, but verification of even small programs requires substantial effort by a team of specialists.
Existing approaches for solver-aided languages~\cite{leino,hamin} typically emphasize generality by supporting shared-memory concurrency with dynamic thread creation, but this generality can make it difficult to reason about the functional behavior of programs.
As a result, prior work typically focuses on relatively simple synchronization patterns, e.g., assuming that all receive operations are blocking and all send operations are not~\cite{leino}.

To address these limitations, we present a new, scalable approach to verifying programs that use the Message Passing Interface (MPI)~\cite{mpi-spec}. 
Our approach supports verification in the context of blocking, non-blocking, and collective communication behaviors. 
This is crucial in the MPI setting, where multiple messages may transfer concurrently and asynchronously relative to the initiating process and large messages can cause a send operation to nondeterministically block when buffering is unavailable.
We are able to support these behaviors at the cost of both requiring the user to specify the communication topology upfront and relying on the fact that MPI disallows dynamic process creation and assumes no shared memory.

We implement our approach as a library called DafnyMPI, for the Dafny programming language~\cite{dafny,dafny2}. 
Dafny is commonly used for the verification of sequential programs~\cite{aws,evm} but it lacks built-in concurrency support.
By integrating DafnyMPI as a library, we show that our approach can be easily integrated with an existing verification framework. 
This also allows us to preserve all the features already available in Dafny, making DafnyMPI more accessible to software engineers with no prior experience verifying parallel software.
Specifically, Dafny already enforces termination and prevents such runtime errors as division by zero.
Dafny automates much of the proof process, and the user guides verification by writing pre- and postconditions and invariants. 
DafnyMPI builds on top of this system by presenting an API that the program can invoke for inter-process communication.
Communication correctness, including deadlock freedom, is guaranteed by DafnyMPI if the program satisfies the preconditions on DafnyMPI method calls, and the postconditions allow reasoning about the return values.
(Section~\ref{sec:overview} introduces Dafny and DafnyMPI and demonstrates their use to reason about a parallelized numerical implementation of the linear convection partial differential equation.)

To prove that the specification of DafnyMPI guarantees deadlock freedom, we introduce a core calculus that closely models the behavior of key MPI primitives: barriers and non-blocking send and receive operations.
An ordering property on message tags ensures that a process blocked on the message with the lowest tag is eventually unblocked by a matching command from its communication partner.
Crucially, our operational semantics models the nondeterministic behavior of MPI send operations: such operations may either block until the receiving process initiates the communication on its end or, if the MPI runtime has enough memory available, the message may be transferred to an internal buffer and the sender may be allowed to proceed.
(Section~\ref{sec:proof} formally defines the core calculus and its operational semantics and outlines the proof of deadlock freedom.)

A user of DafnyMPI may not only want to ensure their program is correct but also to reason about its output.
Since concurrent code is error-prone, it is often useful to prove that the parallel version produces the same result as a simpler sequential one.
DafnyMPI makes this possible by requiring users to specify in advance the expected results of each inter-process communication operation and by enforcing that the read and write buffers are used safely.
Reasoning about functional equivalence between two manually written programs (one parallel and one sequential) can then proceed by transitivity: the user writes a functional specification and proves that both implementations satisfy it.
(Section~\ref{sec:func} describes how we can prove functional equivalence with DafnyMPI.)

To apply DafnyMPI to real-world programs, we use standard Dafny techniques to prove termination and address engineering issues ranging from practical MPI constraints, such as tag overflow, to proof engineering challenges, such as brittleness.
Additionally, we specify the behavior of several useful higher-level collective communication operations, which our core calculus abstracts away as decorated barriers. 
Overall, after dealing with all of the above, DafnyMPI consists of around 2000 lines of Dafny code, of which under 400 lines form the trust base that models our core calculus.
(Section~\ref{sec:engineering} discusses the structure of DafnyMPI and associated engineering considerations.)

To demonstrate the practical utility of DafnyMPI, we use it to verify numerical solutions for three partial differential equations (PDEs): Linear Convection, Heat Diffusion, and Poisson.
We discuss the proof engineering effort required to verify each benchmark and show that the increased complexity of the parallel versions is mirrored by a higher percentage of ghost (non-executable) code.
Finally, we compile the verified Dafny code to Python and demonstrate that the resulting programs exhibit the expected runtime behavior under parallel execution---specifically, that increasing the number of processes leads to faster computation.
(Section~\ref{sec:evaluation} describes these experiments.)

In summary, the main contributions of this paper are as follows:
\begin{itemize}
\item We introduce a technique for verifying concurrent MPI programs and establishing deadlock freedom, termination, and functional equivalence with a sequential implementation.
\item We implement this technique as DafnyMPI, a library for the Dafny programming language.
\item We verify three MPI programs to evaluate the practical benefits of the approach.
We show that the parallel implementations do exhibit better runtime performance than their sequential counterparts.
\end{itemize}

DafnyMPI and the associated benchmarks are available at  \href{https://doi.org/10.5281/zenodo.17102521}{https://doi.org/10.5281/zenodo.17102521} 

In summary, by verifying MPI programs in a language designed for sequential reasoning, we aim to push the boundaries of scalable verification and support the development of reliable software.

\section{Overview}
\label{sec:overview}

\begin{figure}[tb]
	\includegraphics[width=\linewidth]{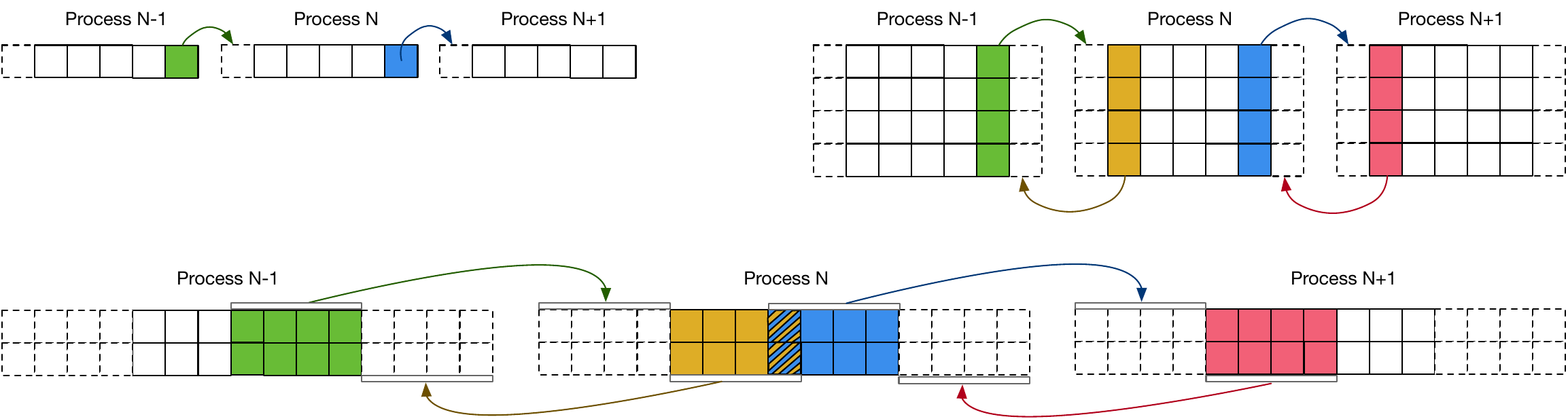}
	\caption{Structure of MPI communication for the Linear Convection (top left), Poisson (top right, transposed), and Heat Diffusion (bottom, transposed) PDE solver implementations.}
	\label{fig:comm}
\end{figure}

MPI is best suited for applications with well-defined communication topologies, where parallelism serves to optimize computation.
One of the key uses of MPI is in partial differential equation (PDE) solvers, where each process is responsible for computing a portion of the solution.
Figure~\ref{fig:comm} shows the structure of MPI communication for three PDE solvers that we use as our benchmarks. 
All three compute the solution iteratively, and each process must exchange data with its neighbors, as indicated by arrows in the figure.
In this section, we focus on the linear convection equation (top left in Figure~\ref{fig:comm}), which is the simplest of the three because each process only sends one value and only receives one value. 
The remaining two benchmarks are discussed in detail in Section~\ref{sec:evaluation}.
Here, we begin by outlining the mathematics behind the PDE to motivate concurrency and then present the corresponding code and discuss our verification methodology.

\subsection{Example: Linear Convection Partial Differential Equation}

The linear convection equation models a wave moving at a constant speed. 
More formally, it describes a quantity $u(x,t)$ that varies in space and time, and it states that the change over time ($\frac{\partial u}{\partial t}$) is proportional to the slope ($\frac{\partial u}{\partial x}$): 

\vspace{-0.5em}
\begin{equation}
\frac{\partial u}{\partial t} + c \frac{\partial u}{\partial x} = 0
\label{eq:lincon}
\end{equation}

To solve the equation, one must be able to find the value of $u$ at a given time step $t$ and point $x$ given the initial state $u_0$.
A numerical solution must approximate the partial derivatives in terms of small discrete steps in time ($\Delta t$) and space ($\Delta x$). 
One approach, known as the \emph{upwind} scheme, is to use the forward difference to approximate the temporal derivative and the backward difference to approximate the spatial derivative:

\vspace{-0.5em}
\begin{equation}
\frac{u(x,t+\Delta t) - u(x, t)}{\Delta t} + c \frac{u(x, t) - u(x - \Delta x, t)}{\Delta x} = 0
\label{eq:intermid}
\end{equation}

Solving Equation~\ref{eq:intermid} for $u(x,t+\Delta t)$, we get:

\vspace{-0.5em}
\begin{equation}
u(x,t+\Delta t) = u(x, t) - c \frac{\Delta t}{\Delta x}(u(x, t) - u(x - \Delta x, t))
\label{eq:upwind}
\end{equation}

In other words, the numerical solution iteratively computes $u(x,t+\Delta t)$ for all points in the spatial domain of $u$.
Note that this computation is local---the value at a point $x$ at time $t+\Delta t$ only depends on the values at points $x$ and $x-\Delta x$ at time $t$. 
The locality of the computation means it can be easily parallelized across several processes, where each process is responsible for computing a portion of the solution. 
However, the processes must also exchange information at each time step to compute the boundary values: if point $x-\Delta x$ is outside its allotted portion of its spatial domain, the process responsible for computing $u(x, t+\Delta t)$ must receive the value $u(x-\Delta x, t)$ from its neighbor.
This exchange of boundary values is shown with arrows between processes in Figure~\ref{fig:comm} (top left).
The Message Passing Interface is ideal for facilitating such structured communication.

\subsection{Solving PDEs with the Message Passing Interface} 

The Python code on the left side of Figure~\ref{fig:example} implements Equation~\ref{eq:upwind} using MPI. 
As is standard in MPI, the code is executed by launching it simultaneously across many CPU cores.
Each process begins by obtaining its rank and the total number of running processes, also known as the size (lines~\ref{line:world}--\ref{line:size}).
Next, each process computes the portion of the initial state it is responsible for (lines~\ref{line:init:start}--\ref{line:init:end}).
The initial conditions can vary; for simplicity, we do not further specify the corresponding function \texttt{init} here. 
The processes with adjacent ranks are responsible for adjacent portions of the spatial domain. 
We will say that process A is to the left of process B if it is responsible for the portion of the spatial domain immediately to B's left on the $x$-axis.

Each process then simulates the equation on \code{u} for \code{nt} time steps (lines~\ref{line:for:start}--\ref{line:barrier}). 
At each time step, the process updates its portion of \code{u} according to Equation~\ref{eq:upwind}, which is implemented by function \code{upwind} in the code (lines~\ref{line:convect:start}--\ref{line:convect:end}). 
Each process performs the computation locally for all points it is responsible for (lines~\ref{line:ifor:start}--\ref{line:ifor:end}) except for the leftmost boundary value (\code{u[0]} in the code), which depends on the information from the neighboring process. 
The leftmost process $0$ only needs to communicate the boundary value to the process on the right (lines~\ref{line:send0:start}--\ref{line:send0:end}), while the rightmost process $N-1$ only needs to receive one from the process on the left. (lines~\ref{line:recvN:start}--\ref{line:recvN:end}).
Every other process both receives a boundary value from the process on the left (lines~\ref{line:recv:start}--\ref{line:recv:end}) and sends one to the process on the right (lines~\ref{line:isend:start}--\ref{line:isend:end}).
Note that for efficiency, the code uses the \emph{immediate send} command (\code{Isend}), which allows the two operations to run concurrently---the immediate send operation is initialized before the receive but is not guaranteed to complete until after the \code{wait} call on line~\ref{line:wait}.
Notice also that each message is tagged with a unique integer derived from the time step and rank. 
We use these tags as a key part of the verification, as discussed below.
The other arguments shared by all send and receive operations are the rank of the destination or source process, respectively, and the buffer---passed by reference---from which the message is read (for sends) or to which it is written (for receives).

A round of communication completes with each process calling \code{Barrier} (line~\ref{line:barrier}), which synchronizes all processes before the next round of communication. 
After \code{nt} iterations, the \code{Gather} command collects all segments of \code{u} into a single array \code{out} managed by process 0 (lines~\ref{line:prep:start}--\ref{line:gather}).
Process 0 may then handle the output as needed, e.g., by saving it to disk or displaying a plot. 

\begin{figure}[tb]
\hspace{0.02\linewidth}\includegraphics[width=0.98\linewidth]{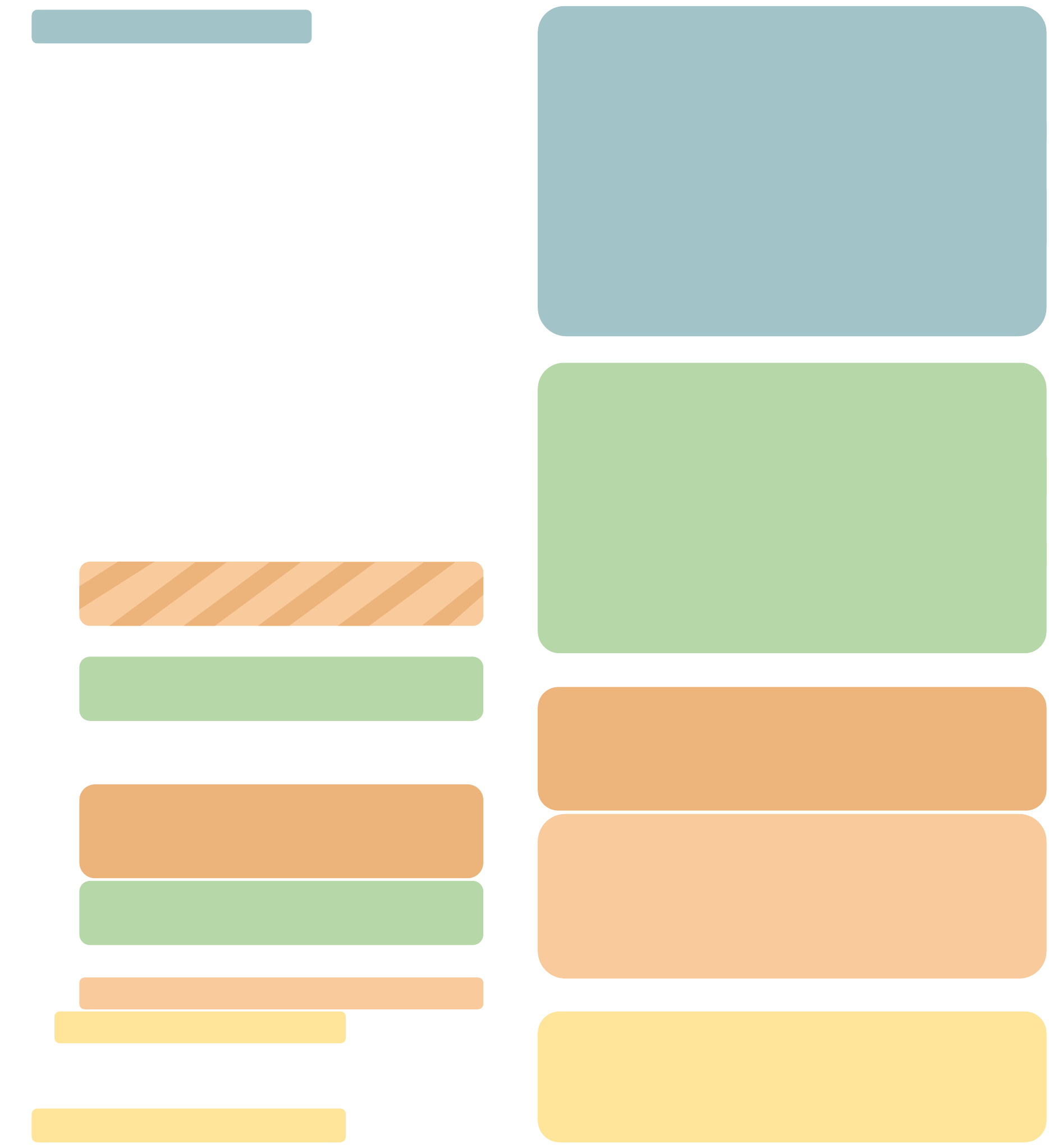}

\vspace{-15cm}
\begin{minipage}{0.48\linewidth}
\lstset{language=python}
\begin{lstlisting}[style=standard,numbers=left,stepnumber=1,firstnumber=auto,name=example]
c = MPI.COMM_WORLD@\label{line:world}@
rank = c.Get_rank()
N    = c.Get_size()@\label{line:size}@
D = MPI.DOUBLE
nt, dx = 25, 0.05@\label{line:init:start}@
# u is the portion of the domain 
# that the process is responsible for
u = init(rank, nx=40, dx=dx, size=N)@\label{line:init:end}@
def upwind(u, u_1, c=1, dt=.025):@\label{line:convect:start}@
  return u - c * dt / dx * (u - u_1)@\label{line:convect:end}@

for n in range(nt):@\label{line:for:start}@
  un = u.copy()
  bt = n * (N - 1)  # base tag
  for i in range(1, len(u)):@\label{line:ifor:start}@
    u[i] = upwind(un[i], un[i - 1])@\label{line:ifor:end}@
  if rank == 0:
    c.Send([un[-1:], D],@\label{line:send0:start}@ 
           dest=rank+1, tag=bt)@\label{line:send0:end}@
  elif rank == N - 1:
    c.Recv([u[0:1], D],@\label{line:recvN:start}@
           src=rank-1, tag=bt+rank-1)@\label{line:recvN:end}@
    u[0] = upwind(un[0], u[0])@\label{line:convect:N}@
  else:
    r = c.Isend([un[-1:], D],@\label{line:isend:start}@
                dest=rank+1,
                tag=bt+rank)@\label{line:isend:end}@
    c.Recv([u[0:1], D],@\label{line:recv:start}@ 
           src=rank-1, tag=bt+rank-1)@\label{line:recv:end}@ 
    u[0] = upwind(un[0], u[0])@\label{line:convect:all}@
    r.wait()@\label{line:wait}@
  c.Barrier()@\label{line:barrier}@
out = zeros(nx, dtype='d') \ @\label{line:prep:start}@
      if rank == 0 else u @\label{line:prep:end}@
c.Gather(u, out, root=0) @\label{line:gather}@
\end{lstlisting} 
\end{minipage}%
\begin{minipage}{0.52\linewidth}
\lstset{language=dafny,
  moredelim=**[is][\color{gray}]{<gray>}{</gray>}}
\begin{lstlisting}[style=standard]
c := new MPI.World(
  N,    // size / # of processes
  Msg,  // (tag, N) -> ...,
  Src,  // (tag, N) -> tag % (N - 1)
  Dest, // (tag, N) -> tag % (N - 1) + 1,
  Clct, // (id,  N) -> if id < nt 
  // then Barrier(tag=id * (N - 1)) 
  // else Gather(tag=nt * (N - 1)...)
  LastClct); // (N) -> nt + 1
tag, clct := -1, 0

assert rank-1 == Src(bt, N) 
 && rank == Dest(bt, N) 
 && NoBetween(tag, bt+rank-1) 
 && bt+rank-1 > tag
 && CanWrite(u[0:1])
 <gray>&& Clct(clct, N).tag <= bt+rank-1 
 && bt+rank-1 < Clct(clct + 1, N).tag...</gray>
tag := bt+rank-1
assume u[0:1] == Msg(bt+rank-1, N)...
 
assert rank == Src(bt, N) 
 && rank+1 == Dest(bt, N) 
 && un[-1:] == Msg(bt, N) 
 && CanRead(un[-1:])
 && NoBetween(tag, bt+rank) 
 && bt+rank > tag
 <gray>&& Clct(clct, N).tag <= bt+rank 
 && bt+rank < Clct(clct + 1, N).tag...</gray>
tag := bt+rank

<gray>assert Clct(clct + 1).tag > tag
 && NoBetween(tag, Clct(clct + 1).tag)...
clct := clct + 1
assume clct==LastClct(N) ==> Done()...</gray>
\end{lstlisting}
\end{minipage}

\caption{Example Python MPI program (left) and corresponding Dafny proof obligations (right, simplified). Proof obligations that deal with collective operations are grayed out to highlight the point-to-point communication.}
\label{fig:example}
\end{figure}

\subsection{DafnyMPI} 
\label{ref:overview:dafny}

The goal of this paper is to develop a methodology for proving that MPI programs such as the one in Figure~\ref{fig:example} are correct, by which we mean that they must terminate successfully with a result that is either identical to that of a matching sequential implementation or equivalent up to rounding differences caused by the non-associative nature of floating-point arithmetic and similar machine-level approximations.
We treat sequential implementations as reference implementations since they are typically easier to write and are therefore less likely to contain bugs. They are also simpler to debug. 
Full program correctness, then, is contingent on the following properties: 
\begin{itemize}     
\item \emph{Deadlock Freedom}. Any process waiting on a send (receive) is eventually unblocked by the matching receive (send), and any process waiting at a barrier eventually passes that barrier. 
We present a formal proof of deadlock freedom in Section~\ref{sec:proof}.
\item \emph{Functional Equivalence}. The MPI program must compute the same result as a corresponding sequential implementation, which is manually written by the programmer and acts as a reference implementation. 
In our benchmarks, each process computes a portion of the final solution, and all such portions are then merged and returned by the root process (rank zero).
DafnyMPI allows reasoning about the output of each process because it enforces the safe use of read and write buffers and ensures that all messages conform to the user-provided specification.
We discuss these properties and the methodology for proving functional equivalence in Section~\ref{sec:func}.
\item \emph{Termination}. Every process will eventually terminate. Aside from deadlock freedom, which we list separately, termination requires the absence of infinite loops, endless recursion, and runtime errors such as division by zero.
Dafny enforces all these properties by default, and so we use the standard techniques for proving them.
Section~\ref{sec:engineering} discusses termination and related proof-engineering considerations such as proof brittleness.
\end{itemize}

\paragraph{Dafny}
We choose to prove these properties in Dafny, a verification-aware language that enforces termination and the absence of runtime errors, including division by zero, out-of-bounds array access, and null pointer dereference.
Dafny automates much of the proof process by applying weakest-precondition calculus to generate verification conditions, which are then discharged to an SMT solver.
The user of Dafny may write pre- and postconditions and invariants to guide the proof process and to enforce additional properties.
In particular, functional equivalence between a pair of methods can be established transitively by requiring both to satisfy the same functional specification in their postconditions.

Thus, given a Python program, we manually translate it into an equivalent Dafny program. 
We then parallelize the computation by rewriting it with DafnyMPI and prove that the resulting version is functionally equivalent to the original.
After verification, both Dafny programs can be compiled back into Python for execution.

The key challenge with this approach is that Dafny is a single-threaded language, with no built-in notion of parallelism or concurrency.
Thus, DafnyMPI is designed with the meta-property that, if the proof obligations of its MPI primitives are met, then the parallel execution of the program is guaranteed to satisfy the properties listed above.

\paragraph{DafnyMPI proof obligations for point-to-point communication}
The right side of Figure~\ref{fig:example} shows a selection of the DafnyMPI proof obligations, written as Dafny code, for our example, which we discuss next.
Here and in Section~\ref{sec:proof}, we focus on deadlock freedom, as that is the most intricate property to prove.
In the figure, the correspondence between source code and proof obligations is shown by color.
In one case, a program statement corresponds to multiple proof obligations, which is shown with appropriately colored stripes.
We gray out proof obligations that deal with collective operations to keep the focus on the point-to-point communication for now; we return to collective operations later in this section.

At the start of the program, corresponding to line~\ref{line:world} on the left, the code creates a new \code{MPI.World} instance that must include programmer-supplied specifications of the overall communication pattern of the program.
Specifically, the programmer provides the \emph{size} \code{N}, the total number of processes---which, as we do here, can be left unconstrained to show the proof holds for any number of processes greater than one---and a list of functions describing the program's use of point-to-point messages and collective operations.
The first three functions, \code{Msg}, \code{Src}, and \code{Dest}, describe point-to-point messages, mapping message tags and the size to the contents of the message, the message source, and the message destination, respectively.
For our example, we omit the message contents for conciseness, and the source and destination are computed from the tag and the size.
The next two functions, \code{Clct} and \code{LastClct}, describe the program's collective operations, including barriers.
\code{LastClct} specifies the total number of collective operations in the program, which, in the general case, is a function of the number of processes.
In our example, there are always exactly \code{nt} barriers---one per loop iteration (lines~\ref{line:for:start}--\ref{line:barrier})---and one \code{Gather} call on line~\ref{line:gather}, so \code{LastClct} is constant and is set to \code{nt+1}.
The \code{Clct} function maps each collective operation ID to its description, including the operation type (e.g., \code{Barrier} or \code{Gather}) and the lowest tag---across all processes---of any point-to-point message that immediately follows it.
This allows us to order collective operations, which otherwise cannot be assigned tags, relative to point-to-point messages.
Finally, the variable \code{tag} is initialized to $-1$ to track the tag of the last completed message, and \code{clct} is initialized to $0$ to track the ID of the next upcoming collective operation.

Next, the code corresponding to the two \code{Recv} calls on lines~\ref{line:recvN:start}--\ref{line:recvN:end}~and~\ref{line:recv:start}--\ref{line:recv:end} asserts that, in order,
the source and destination of the message match the tag, 
no message was skipped (the \code{NoBetween(a,b)} predicate holds if and only if there is no tag between \code{a} and \code{b} for which the current process is the sender or the receiver),
the current tag is greater than the one previously used,
the destination buffer (\code{u[0:1]}) can be written to,
and the tag falls between the previous and the next collective operations (in our case, barriers).
Once the command executes, the process's \code{tag} counter is updated to the message's tag. 
Finally, the snippet assumes (the rely part of rely--guarantee reasoning) that \code{recv} returns the message \code{Msg(bt+rank-1,N)} corresponding to the tag in the specification.

The code corresponding to the \code{Send} call on lines \ref{line:send0:start}--\ref{line:send0:end} similarly begins by asserting that the source, destination, and contents of the message match the tag (the guarantee part of rely--guarantee reasoning), and that the source buffer is free to be read from.
Next, the code enforces the same tag ordering properties as for the \code{Recv} call.
After the send completes, the process's \code{tag} counter is incremented.
Whenever a send is non-blocking, such as on lines~\ref{line:isend:start}--\ref{line:wait}, these different conditions are enforced at different points along the program's execution.
The ordering properties on the tag are enforced at the corresponding \code{wait} call (line~\ref{line:wait}) that completes the send, while the rest of the conditions are checked immediately when the point-to-point communication is initialized (lines~\ref{line:isend:start}--\ref{line:isend:end}).
Additional pre- and postconditions, not shown in the figure, ensure that the process does not send a duplicate message and that the send buffer is not written to between the calls to \code{ISend} and \code{wait}.

To understand why the conditions listed above prevent deadlock, suppose, for a contradiction, that deadlock does occur and that, among all the blocked processes, process $P$ is blocked on an MPI call with the lowest tag.
The matching process for that message must not have initiated the communication with $P$ on its own end yet, since otherwise $P$ would eventually unblock.
Therefore, the matching process must not have reached the point in the program at which it initiates communication with $P$.
At the same time, there is a contradiction, because the matching process cannot be blocked at this point in the program, since then it would be the process blocked on the message with the lowest tag. 
Therefore, the matching process must be running and so there is no deadlock. (Section~\ref{sec:proof} formalizes this argument.)

\paragraph{Collective operations.}
The final block of code on the right side of Figure~\ref{fig:example} deals with collective operations and corresponds to the \code{Barrier} call on line~\ref{line:barrier} and the \code{Gather} call on line~\ref{line:gather}. 
The code asserts that no barrier or collective operation was skipped and that the barrier is called only after any message that must precede it.
The ordering of collective operations relative to point-to-point messages is additionally enforced on every call to \code{Recv}, \code{Send}, or \code{wait}, as shown by the grayed-out proof obligations in the relevant code snippets.

At each call to \code{Barrier} or \code{Gather}, the ID of the next collective operation is incremented.
When this ID reaches \code{LastClct(N)}, Dafny can assume \code{Done()}, a special condition indicating that the program has completed its execution.
The user of the DafnyMPI library must have a postcondition on the main method that checks that \code{Done()} is true. 
In our example, this condition is only fulfilled after the call to \code{Gather} on line~\ref{line:gather}.
This call has additional pre- and postconditions attached to it (not shown in the figure) that ensure the inputs conform to the specification provided by the user, e.g., that \code{u} holds part of the output corresponding to the process's \code{rank}.

Together, these conditions guarantee deadlock freedom and allow the user to prove that the variable \code{out} will, for the zeroth process, contain the solution to the PDE after \code{nt} time steps.
The next section formalizes this approach for a core calculus where we prove that these conditions are sufficient. 

\section{Formalizing DafnyMPI}
\label{sec:proof}

In this section, we formalize DafnyMPI on a core calculus and prove that for a given program, if DafnyMPI's requirements are satisfied, then that program will be deadlock-free when run in parallel.
Our core calculus adds several primitives from the MPI~1.0 standard~\cite{mpi-spec} to a language of commands.
As far as we are aware, later MPI standards maintain the semantics of these primitives, so our approach would still apply.
Moreover, in our experience, most MPI programs are written in a way that does not require the more advanced features of later standards, such as MPI~2.0 or MPI~3.0, which introduce additional features such as one-sided communication and dynamic process management.

DafnyMPI directly encodes the following four MPI functions:
\begin{itemize}
	\item \texttt{MPI\_IRECV} (\emph{immediate receive} or \emph{non-blocking receive}) requests to receive a message with a given tag from a given source. 
    The operation returns immediately, and the message may then be received in the background.
    The operation returns an \texttt{MPI\_REQUEST} object that can be waited on to ensure that the message has been received. 
    \item \texttt{MPI\_ISEND} (\emph{immediate send} or \emph{non-blocking send}) requests to send a message with a given tag to a given destination. 
    The operation returns immediately with an \texttt{MPI\_REQUEST} object that can be waited on to ensure the send has been completed. 
    Note that completion is local in this context: it does not guarantee that the message is received at the destination, only that the sender can reuse the buffers associated with the operation.
    \item \texttt{MPI\_WAIT} (\emph{wait}), called on an  \texttt{MPI\_REQUEST} object, blocks a process until the corresponding communication has been completed.  
    \item \texttt{MPI\_BARRIER} (\emph{barrier}) blocks until all processes also reach a barrier.	
\end{itemize}

In Section~\ref{sec:engineering}, we discuss how we can extend the core calculus
to support MPI functions that are available in DafnyMPI, using syntactic sugar
or as simple variations of the initial four.
\begin{itemize} 
    \item \texttt{MPI\_RECV} (\emph{blocking receive}) can be encoded as an \texttt{MPI\_IRECV} followed immediately by the relevant \texttt{MPI\_WAIT}. 
    \item \texttt{MPI\_SEND} (\emph{blocking send}) can be encoded similarly with \texttt{MPI\_ISEND} and \texttt{MPI\_WAIT}.
    \item \texttt{MPI\_GATHER} must be called by all processes, after which the data associated with each process is transferred to the designated root process and stored in a single array. 
We encode this operation as a decorated barrier call, requiring that all processes are synced before exchanging information in this way.
   \item \texttt{MPI\_ALLREDUCE} must be called by all processes, after which the data submitted by each process is folded using a binary operation such as summation or logical and.
   	The result is then communicated back to all processes. 
   	This operation is encoded as a decorated barrier similarly to \texttt{MPI\_GATHER}.
   	\end{itemize}

DafnyMPI does not currently support MPI \emph{communicators}, which group processes together so that, e.g., only a group needs to pass a barrier, or \emph{wildcards}, which allow one process to exchange a message with an arbitrary process associated with the same communicator.
We leave support for these features to future work.
Instead, DafnyMPI programs use the single default \texttt{MPI\_COMM\_WORLD} communicator on all operations.

\subsection{The Core Calculus}
\label{sec:proof:core}

Figure~\ref{fig:grammar} shows the syntax of our core calculus, which is stratified into expressions and commands.
As our proof focuses on interprocess communication, expressions~$e$ are limited to integers~$n$, variables~$x$, comparisons $e~\textbf{=}~e$, and basic arithmetic operations $e~\textbf{+}~e$, $e~\textbf{*}~e$, $e~\textbf{div}~e$, $\textbf{-}e$.
Because there is no shared memory, each process has its own variable environment, and there is no way for one process to affect the evaluation of an expression by another process.
However, an expression may not read a variable that is being written by a background receive call (see Section~\ref{sec:proof:axioms}).

Commands~$c$ include four communication primitives corresponding to the MPI functions listed above: $\textbf{irecv}~e~x$, non-blocking receive of a message with tag $e$, where the message is written to $x$; $\textbf{isend}~e~x$, non-blocking send of a message with tag $e$, where the message is read from $x$; $\textbf{wait}~e$, waiting for completion of the send or receive  with tag $e$; and $\textbf{barrier}$, which blocks until all processes reach a barrier.
In all cases, tags are integers.
The remaining commands---\textbf{skip}, \textbf{if}, \textbf{while}, \textbf{set}, and sequencing---have the standard semantics, where zero is considered false and non-zero integers are considered true.

\begin{wrapfigure}[10]{r}{0.35\linewidth}
\vspace{-1em}
\centering
\begin{minipage}{\linewidth}
\begin{displaymath}
\begin{array}{lll}
e &   ::=  & n
      \mid   x
      \mid   e~\textbf{=}~e 
      \mid   e~\textbf{+}~e 
      \mid   \textbf{-}e \\
 &    \mid & e~\textbf{*}~e
      \mid   e~\textbf{div}~e \\
c &   ::=  & \textbf{irecv}~e~x
      \mid   \textbf{isend}~e~x \\
    & \mid & \textbf{wait}~e
      \mid   \textbf{barrier} \\
    & \mid & \textbf{skip}
      \mid   \textbf{if}~e~c~c \\
    & \mid & \textbf{while}~e~c
      \mid   \textbf{set}~x~e
      \mid   c~\textbf{;}~c
\end{array}
\end{displaymath}
\caption{Core calculus syntax.} 
\label{fig:grammar}
\end{minipage}
\end{wrapfigure}

Our semantics models the state of a single process as a tuple $\langle c, \sigma, t, b\rangle$, where $c$ is the command to be executed, $\sigma$ is the variable environment mapping variables to integers, $t$ is the largest tag previously used in a \textbf{wait} call, and $b$ is the number of barriers the process has passed so far.
Then the global program state $S$, encompassing all processes, is a tuple $(\mathcal{P}, \mathcal{B}_r, \mathcal{B}_s, \mathcal{B}_m)$, where $\mathcal{P}$ is a set of process states, and $\mathcal{B}_r$, $\mathcal{B}_s$, and $\mathcal{B}_m$ are the global receive, send, and message buffers, respectively.
Throughout this section, we will also use $N=|\mathcal{P}|$ to refer to the number of processes.

The receive and send buffers are maps from message tags to variable names, indicating the destination (for receives) or source (for sends) of the message within the process environment. 
The message buffer is a map from message tags to message payload, which is the value that would eventually be copied from the source variable in the sender's environment to the destination variable in the receiver's environment.
This intermediate payload buffer is necessary because MPI allows a send operation to complete before the corresponding receive is called, in which case the message contents must be buffered internally, while the sender is allowed to reuse the source variable.
It is sufficient to model the buffers globally because the programmer is required to provide functions that uniquely identify the sender and receiver for each message tag as part of the program's specification. 
Section~\ref{sec:proof:opsem} goes through a concrete example of how the buffers are used.

Our operational semantics is nondeterministic and may take a step in any process that is not blocked.
Thus, we treat the running processes as a set and use the set notation for convenience.
In particular, we write $\langle c, \sigma, t, b\rangle \cup \mathcal{P}$ (omitting curly braces around the singleton set) to select one process $\langle c, \sigma, t, b\rangle$ with the remaining processes as $\mathcal{P}$, and we write $\mathcal{P}_1 \cup \mathcal{P}_2 \cup \ldots$ to nondeterministically split the set of processes into subsets.
We store the process rank by binding an integer to \texttt{rank} in $\sigma$.
As a shorthand, we use a subscript on the process state to indicate rank, i.e., $\langle c, \sigma, t, b\rangle_i$ indicates a state where $\sigma(\texttt{rank}) = i$.
We omit the subscript when the rank is irrelevant.

A state reduction is a sequence of states $S_0, S_1, \ldots, S_k$ such that each preceding state transitions to the next one according to the operational semantics (Section~\ref{sec:proof:opsem}).
As another shorthand, we sometimes refer to the state components of the $i$th process during the $k$th reduction step by subscripting with $i,k$, e.g., $\sigma_{2,0}$ refers to the variable environment of the process with rank $2$ in the initial state $S_0$.

DafnyMPI programs are \emph{single instruction, multiple data} (SIMD), with the same command $c$ executed on all nodes.
The initial state of the execution of command $c$ on $N$ nodes is given by:
\[
S_0 = (\bigcup_{0\leq i < N}\langle c;~\textbf{barrier}, \{\texttt{rank}\mapsto i, \texttt{size}\mapsto N\}, -1, 0\rangle, \emptyset, \emptyset,\emptyset)
\]
Notice the processes are initialized with variables \texttt{rank} and \texttt{size} to let the process know its rank and the total number $N$ of processes running.
Each process is also initialized with the exact same sequence of commands to execute, where the last command must be a \textbf{barrier}.

We verify a given program with respect to a user-provided specification, where, as in Section~\ref{sec:overview}, the user must describe the overall communication pattern of the program ahead of time. 
Specifically, the user provides five functions below, which we will refer to throughout this section and which together describe all MPI operations.
Note that all five functions are parameterized only by process-independent quantities, so any two processes will always compute the same values:
\begin{itemize}
	\item \textsc{Sender} and \textsc{Receiver} are total functions that map a message tag and the number of processes to the rank of the process that sends or receives a message with that tag. 
	In our running example in Figure~\ref{fig:example}, these correspond to functions \code{Src} and \code{Dest}, respectively.
	\item \textsc{Message} (\code{Msg} in Figure~\ref{fig:example}) is a total function mapping a message tag and the number of processes to the message payload, which in our core calculus semantics may only be an integer. 
	\item \textsc{BarrierCount} (\code{LastClct} in Figure~\ref{fig:example}) maps the number of processes to the total number of barriers a process must pass before terminating.
	\item \textsc{Barrier} is a total function that maps the rank of a barrier and the number of processes to a message tag, indicating that any message with a smaller tag must be sent and received before the process reaches the barrier.
	This function roughly corresponds to function \code{Clct} in Figure~\ref{fig:example}, although DafnyMPI supports collective operations not explicitly modeled by the core calculus, as we describe in Section~\ref{sec:engineering}.
\end{itemize}

%Next, we present the operational semantics of our core calculus and the guarantees that are provided by Dafny. 

\subsection{Operational Semantics}
\label{sec:proof:opsem}

This section presents the small-step operational semantics of our core calculus.
The goal is to model the behavior of message-passing programs as closely as possible to the MPI standard, so that the proof of deadlock freedom that follows applies in all possible situations.
One design choice that distinguishes our semantics from much of prior work is that we allow send operations to block nondeterministically.
In prior work (e.g.,~\cite{leino}), send operations are typically assumed to be non-blocking, which simplifies reasoning but does not capture the scenario, easily triggered by large message payloads, in which the MPI implementation temporarily runs out of internal buffer space.
To model this behavior, our semantics allows a message to be transferred from the sender to the shared message buffer $\mathcal{B}_m$ in one of two ways:
(i) eagerly, before the sender reaches a \textbf{wait} call (representing a normal immediate send), or
(ii) lazily, only when the receiver executes its own \textbf{wait} (representing an immediate send that blocks).
The remainder of this section describes these rules.

Figure~\ref{fig:reductions} gives the small-step operational semantics, divided into two groups: (\subref{fig:reductions:single}) local rules that touch a single process, and (\subref{fig:reductions:multi}) global rules that also manipulate shared buffers or several processes at once.
Our semantics also includes big-step expression evaluation, $\langle e,\sigma\rangle \Downarrow v$, meaning evaluating expression $e$ in variable environment $\sigma$ yields value $v$.
As expression evaluation is entirely standard, we omit these rules. 

The local rules, shown in Figure~\ref{fig:reductions:single}, define the single stepping relation $\rightsquigarrow$ on process states $\langle c,\sigma, t, b\rangle$ meaning taking one step in the process whose state is on the left yields the state on the right.
\textsc{IfTrue} and \textsc{IfFalse} cover the two cases for conditionals, and \textsc{While} expands one step of the iteration.
Finally, \textsc{Set} updates the variable environment appropriately, and \textsf{SeqSkip} removes a \textbf{skip} at the start of a sequence.
A process whose command is a lone \textbf{skip} is considered to have terminated, and thus there is no rule to evaluate it further.

\begin{figure}
\begin{subfigure}{\textwidth}
\begin{mathpar}
\small
\inferrule*[right=IfTrue]
{\langle e, \sigma \rangle \Downarrow v\quad v\neq 0}
{\langle \textbf{if}~e~c_1~c_2,\sigma,t,b\rangle
\rightsquigarrow
\langle c_1,\sigma,t,b\rangle}

\inferrule*[right=IfFalse]
{\langle e, \sigma \rangle \Downarrow 0}
{\langle \textbf{if}~e~c_1~c_2,\sigma,t,b\rangle
\rightsquigarrow
\langle c_2,\sigma,t,b\rangle}

\inferrule*[right=While]
{ }
{\langle \textbf{while}~e~c,\sigma,t,b\rangle
\rightsquigarrow
\langle \textbf{if}~e~(c;~\textbf{while}~e~c)~\textbf{skip},\sigma,t,b\rangle}

\inferrule*[right=Set]
{\langle e, \sigma \rangle \Downarrow v}
{\langle\textbf{set}~x~e,\sigma,t,b\rangle
\rightsquigarrow
\langle\textbf{skip},\sigma[x\mapsto v],t,b\rangle}

\inferrule*[right=SeqSkip]
{ }
{\langle \textbf{skip};~c,\sigma,t,b\rangle
\rightsquigarrow
\langle c,\sigma,t,b\rangle}
\end{mathpar}
\caption{Single-Process Operational Semantics Rules.}
\label{fig:reductions:single}
\end{subfigure}

\bigskip{}

\begin{subfigure}{\textwidth}
\begin{mathpar}
\small

\inferrule*[right=Proc]
{\langle c,\sigma,t,b\rangle
\rightsquigarrow
\langle c',\sigma',t',b'\rangle}
{(\langle c,\sigma,t,b\rangle\cup\mathcal P,\mathcal B_r, \mathcal B_s, \mathcal B_m)
\longrightarrow
(\langle c',\sigma',t',b'\rangle\cup\mathcal P,\mathcal B_r, \mathcal B_s, \mathcal B_m)}

\inferrule*[right=SeqStep]
{(\langle c_1,\sigma,t,b\rangle \cup \mathcal P, \mathcal B_r, \mathcal B_s, \mathcal B_m)
\longrightarrow
(\langle c',\sigma',t',b'\rangle \cup \mathcal P, \mathcal B'_r, \mathcal B'_s, \mathcal B'_m)}
{(\langle c_1\textbf{; }c_2,\sigma,t,b\rangle \cup \mathcal P, \mathcal B_r, \mathcal B_s, \mathcal B_m)
\longrightarrow
(\langle c'\textbf{; }c_2,\sigma',t',b'\rangle \cup \mathcal P, \mathcal B'_r, \mathcal B'_s, \mathcal B'_m)}

\inferrule*[right=Send]
{\langle e,\sigma\rangle\Downarrow t'}
{(\langle\textbf{isend}~e~x,\sigma,t,b\rangle\cup\mathcal P,\mathcal B_r, \mathcal B_s, \mathcal B_m)
\longrightarrow
(\langle\textbf{skip},\sigma,t,b\rangle\cup\mathcal P,\mathcal B_r, \mathcal B_s[t'\mapsto x], \mathcal B_m)}

\inferrule*[right=Recv]
{\langle e, \sigma \rangle \Downarrow t'}
{(\langle\textbf{irecv}~e~x,\sigma,t,b\rangle\cup\mathcal P,\mathcal B_r, \mathcal B_s, \mathcal B_m)
\longrightarrow
(\langle\textbf{skip},\sigma,t,b\rangle\cup\mathcal P,\mathcal B_r[t'\mapsto x], \mathcal B_s, \mathcal B_m)}

\inferrule*[right=WaitRecv]
{\langle e, \sigma \rangle \Downarrow t'\quad
t'\in \mathcal B_r\quad
t' \in \mathcal B_m\quad}
{(\langle\textbf{wait}~e,\sigma,t,b\rangle_i\cup\mathcal P,\mathcal B_r, \mathcal B_s, \mathcal B_m)
\longrightarrow
(\langle\textbf{skip},\sigma[\mathcal B_r(t')\!\mapsto\!\mathcal B_m(t')],t',b\rangle_i\cup\mathcal P,\mathcal B_r\!\setminus\!\{t'\}, \mathcal B_s, \mathcal B_m)}

\inferrule*[right=WaitSend]
{\langle e, \sigma \rangle \Downarrow t'\quad
t' \in \mathcal{B}_m}
{(\langle\textbf{wait}~e,\sigma,t,b\rangle_i\cup\mathcal P,\mathcal B_r, \mathcal B_s, \mathcal B_m)
\longrightarrow
(\langle\textbf{skip},\sigma,t',b\rangle_i\cup\mathcal P,\mathcal B_r, \mathcal B_s\!\setminus\!\{t'\}, \mathcal B_m)}

\inferrule*[right=TransferOnWait]
{
\langle e, \sigma \rangle \Downarrow t' \quad
\textsc{Receiver}(t',N) = i\quad
t' \in \mathcal B_s\quad 
\mathcal B_s(t') = x\quad 
\langle x,\sigma_{\textsc{Sender}(t',N)} \rangle \Downarrow v}
{(\langle \textbf{wait}~e, \sigma, t, b\rangle_i \cup \mathcal P,\mathcal B_r, \mathcal B_s, \mathcal B_m)
\longrightarrow
(\langle \textbf{wait}~e, \sigma, t, b\rangle_i \cup \mathcal P,\mathcal B_r, \mathcal B_s, \mathcal B_m[t'\mapsto v])}

\inferrule*[right=TransferNoWait]
{
t' \in \mathcal B_s\quad 
i=\textsc{Sender}(t',N)\quad
c_i \neq \textbf{wait}~e\textbf{;~}c'\quad
c_i \neq \textbf{barrier}\textbf{;~}c'\quad
\mathcal B_s(t') = x\quad 
\langle x,\sigma_i \rangle \Downarrow v}
{(\mathcal P,\mathcal B_r, \mathcal B_s, \mathcal B_m)
\longrightarrow
(\mathcal P,\mathcal B_r, \mathcal B_s, \mathcal B_m[t'\mapsto v])}

\inferrule*[right=FreeBuffer]
{i = \textsc{Sender}(t',N)\quad
j = \textsc{Receiver}(t',N)\quad
t_i \geq t'\quad 
t_j \geq t'\quad
t' \in \mathcal{B}_m}
{(\mathcal P,\mathcal B_r, \mathcal B_s, \mathcal B_m)
\longrightarrow
(\mathcal P,\mathcal B_r, \mathcal B_s, \mathcal B_m\!\setminus\!\{t'\})}

\inferrule*[right=Barrier]
{~}
{(\bigcup_{0 \leq i < N}\langle\textbf{barrier}\textbf{; }c,\sigma,t,b\rangle_i,\mathcal B_r, \mathcal B_s, \mathcal B_m)
\longrightarrow
(\bigcup_{0\leq i < N}\langle\textbf{skip; }c,\sigma,t,b+1\rangle_i,\mathcal B_r, \mathcal B_s, \mathcal B_m)}
\end{mathpar} 
\caption{Multi-Process Operational Semantics Rules.}
\label{fig:reductions:multi}
\end{subfigure}
\caption{Core calculus reduction rules.}
\label{fig:reductions}
\end{figure}

\begin{figure}[bt]
\tikzset{
  statebox/.style={
    draw,
    rectangle,
    rounded corners=3pt,
    inner sep=2pt,
    text width=1.68cm,
    font=\tiny,
    align=left
  }
}

\begin{tikzpicture}
\matrix[matrix of nodes,
        nodes={statebox},
        column sep=0.55cm,
        row sep=0.2cm,
        nodes in empty cells=false
       ] (M) {

  % Row 1
  \parbox{1.68cm}{\centering \bfseries State 0\\[-5pt]
  \rule{\linewidth}{0.4pt}\\[0pt]
  \raggedright
  $c_0 = \textbf{isend}~t~x\textbf{;}..$\\
  $c_1 = \textbf{irecv}~t~x\textbf{;}..$\\
  $\mathcal B_r = \emptyset$\\
  $\mathcal B_s = \emptyset$\\
  $\mathcal B_m = \emptyset$} &

  \parbox{1.68cm}{\centering \bfseries State 1\\[-5pt]
  \rule{\linewidth}{0.4pt}\\[0pt]
  \raggedright
  $c_0 = \textbf{skip}\textbf{;}~\textbf{wait}~t$\\
  $c_1 = \textbf{irecv}~t~x\textbf{;}..$\\
  $\mathcal B_r = \emptyset$\\
  $\mathcal B_s = \{t\mapsto x\}$\\
  $\mathcal B_m = \emptyset$} &

  \parbox{1.68cm}{\centering \bfseries State 2\\[-5pt]
  \rule{\linewidth}{0.4pt}\\[0pt]
  \raggedright
  $c_0 = \textbf{skip}\textbf{;}~\textbf{wait}~t$\\
  $c_1 = \textbf{irecv}~t~x\textbf{;}..$\\
  $\mathcal B_r = \emptyset$\\
  $\mathcal B_s = \{t\mapsto x\}$\\
  $\mathcal B_m = \{t\mapsto 5\}$} &

  \parbox{1.68cm}{\centering \bfseries State 3\\[-5pt]
  \rule{\linewidth}{0.4pt}\\[0pt]
  \raggedright
  $c_0 = \textbf{skip}$\\
  $c_1 = \textbf{irecv}~t~x\textbf{;}..$\\
  $\mathcal B_r = \emptyset$\\
  $\mathcal B_s = \emptyset$\\
  $\mathcal B_m = \{t\mapsto 5\}$} &

  \parbox{1.68cm}{\centering \bfseries State 4\\[-5pt]
  \rule{\linewidth}{0.4pt}\\[0pt]
  \raggedright
  $c_0 = \textbf{skip}$\\
  $c_1 = \textbf{wait}~t$\\
  $\mathcal B_r = \{t\mapsto x\}$\\
  $\mathcal B_s = \emptyset$\\
  $\mathcal B_m = \{t\mapsto 5\}$} &

  \parbox{1.68cm}{\centering \bfseries State 5\\[-5pt]
  \rule{\linewidth}{0.4pt}\\[0pt]
  \raggedright
  $c_0 = \textbf{skip}$\\
  $c_1 = \textbf{skip}$\\
  $\mathcal B_r = \emptyset$\\
  $\mathcal B_s = \emptyset$\\
  $\mathcal B_m = \{t\mapsto 5\}$}
  \\
& 
  \parbox{1.68cm}{\centering \bfseries State 6\\[-5pt]
  \rule{\linewidth}{0.4pt}\\[0pt]
  \raggedright
  $c_0 = \textbf{isend}~t~x\textbf{;}..$\\
  $c_1 = \textbf{skip;}~\textbf{wait}~t$\\
  $\mathcal B_r = \{t\mapsto x\}$\\
  $\mathcal B_s = \emptyset$\\
  $\mathcal B_m = \emptyset$} &
\parbox{1.68cm}{\centering \bfseries State 7\\[-5pt]
  \rule{\linewidth}{0.4pt}\\[0pt]
  \raggedright
  $c_0 = \textbf{wait}~t$\\
  $c_1 = \textbf{skip;}~\textbf{wait}~t$\\
  $\mathcal B_r = \{t\mapsto x\}$\\
  $\mathcal B_s = \{t\mapsto x\}$\\
  $\mathcal B_m = \emptyset$} &
\parbox{1.68cm}{\centering \bfseries State 8\\[-5pt]
  \rule{\linewidth}{0.4pt}\\[0pt]
  \raggedright
  $c_0 = \textbf{wait}~t$\\
  $c_1 = \textbf{wait}~t$\\
  $\mathcal B_r = \{t\mapsto x\}$\\
  $\mathcal B_s = \{t\mapsto x\}$\\
  $\mathcal B_m = \{t\mapsto 5\}$} &
\parbox{1.68cm}{\centering \bfseries State 9\\[-5pt]
  \rule{\linewidth}{0.4pt}\\[0pt]
  \raggedright
  $c_0 = \textbf{wait}~t$\\
  $c_1 = \textbf{skip}$\\
  $\mathcal B_r = \emptyset$\\
  $\mathcal B_s = \{t\mapsto x\}$\\
  $\mathcal B_m = \{t\mapsto 5\}$} &  \\\\
};

\draw[->] (M-1-1.east) -- node[midway, above]{\tiny S} (M-1-2.west);
\draw[->] (M-1-2.east) -- node[midway, above]{\tiny TNW} (M-1-3.west);
\draw[->] (M-1-3.east) -- node[midway, above]{\tiny SS}node[midway, below]{\tiny WS} (M-1-4.west);
\draw[->] (M-1-4.east) -- node[midway, above]{\tiny R}node[midway, below]{\tiny SS} (M-1-5.west);
\draw[->] (M-1-5.east) -- node[midway, above]{\tiny WR} (M-1-6.west);
\draw[->] (M-2-2.east) -- node[midway, above]{\tiny S} node[midway, below]{\tiny SS} (M-2-3.west);
\draw[->] (M-2-3.east) -- node[midway, above]{\tiny SS} node[midway, below]{\tiny TOW} (M-2-4.west);
\draw[->] (M-2-4.east) -- node[midway, above]{\tiny WR} (M-2-5.west);
\draw[->] 
  (M-1-1.south) 
  -- ($ (M-1-1.center |- M-2-2.center) $) 
  -- node[midway,above]{\tiny R} (M-2-2.west);
\draw[->] 
  (M-2-5.east) 
  -- node[midway,above]{\tiny WS}
  ($ (M-1-6.center |- M-2-5.center) $) 
  -- (M-1-6.south);
\end{tikzpicture}
\caption{Two partial reductions for the program $\textbf{set}~x~5\textbf{; if}~\texttt{rank}=0~\textbf{\{}~\textbf{isend}~t~x~\textbf{\}~\{}~\textbf{irecv}~t~x~\textbf{\};}~\textbf{wait}~t$. The program is executed by two processes; edges are labeled with abbreviations of reduction rules (SS stands for \textsc{SeqSkip}, etc.). The reduction rule above the line is taken first.}
\label{fig:buffers}
\end{figure}
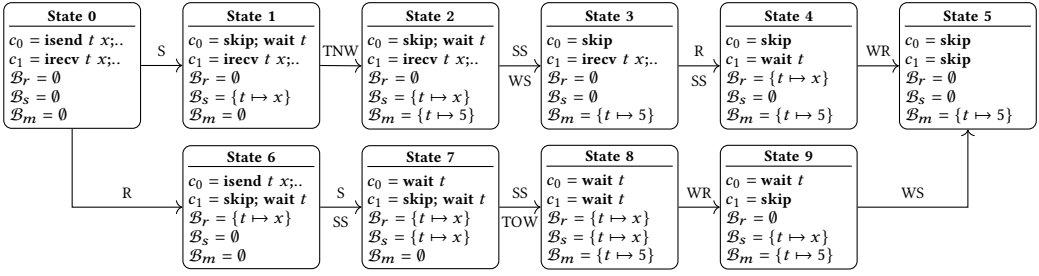

The global rules, shown in Figure~\ref{fig:reductions:multi}, define the single stepping relation $\longrightarrow$ on program states $(\mathcal P,\mathcal B_r, \mathcal B_s, \mathcal B_m)$.
Importantly, these rules omit several premises that would typically be checked dynamically but instead are guaranteed to hold by the Dafny verifier.
For example, the \textsc{Send} rule does not verify that the sender's rank matches the tag's designated sender, because this is checked as part of using DafnyMPI.
We formalize these preconditions as \emph{axioms} in Section~\ref{sec:proof:axioms}, below.

The first two reduction rules are administrative: \textsc{Proc} lifts a single-process reduction to the reduction of the entire program state, and \textsc{SeqStep} takes a step in the first command of a sequence.
Notice that both rules use set notation to pick a process nondeterministically.

Most of the remaining rules model send operations, and the main challenge lies in their nondeterministic behavior.
Depending on the implementation and message size, the MPI runtime may copy the payload into an internal buffer, allowing the sender's subsequent \textbf{wait} to complete immediately.
Alternatively, the process may block until a matching receive is called.
Our operational semantics aims to mimic both of these behaviors.

We use the example in Figure~\ref{fig:buffers} to illustrate the communication rules as we describe them.
The figure shows two of the possible partial reduction sequences of a simple program.
In the initial state (State 0 in the figure), the process with rank~0 is poised to initiate a non-blocking send for a message with tag $t$, and the process with rank~1 is poised to initiate a non-blocking receive for the same message.
Either command may fire first, leading to State~1 or State~6, respectively.

Rule \textsc{Send} adds a request to send $x$ as a message under the appropriate tag $t'$.
The request is registered by placing $x$ in the send buffer $\mathcal{B}_s$.
Similarly, \textsc{Recv} adds a request to receive $x$ with tag $t'$.
As described above, these operations are non-blocking.
In our example, we label the state transitions with \textsc{S} (send) and \textsc{R} (receive), and we see the appropriate updates to $\mathcal B_s$ and $\mathcal B_r$, respectively.

Continuing along the top of the figure, the transition from State~1 to State~2 illustrates \textsc{TransferNoWait}, which, given a sender process that is not waiting or at a barrier (recall $c_i$ indicates the command of the $i$th process), copies the contents of the to-be-sent message from $x$ to $\mathcal{B}_m$ with the appropriate tag $t'$.
In the example, the message payload is 5, and thus $\mathcal{B}_m$ is updated to map $t$ to 5.
This behavior captures the case when the MPI implementation copies a message payload into a buffer, allowing the sender process to unblock before the matching receive is even initiated.

Once the message is copied into the message payload buffer, the process with rank~0 is now ready to complete the send, as captured by the transition from State~2 to State~3 via \textsc{WaitSend}.
The rule checks that the message with tag $t'$ has been placed in $\mathcal{B}_m$ and then removes the tag from the set of pending sends.
In our example, we see that this rule can fire because of the previous execution of \textsc{TransferNoWait}, and therefore tag $t$ is removed from $\mathcal{B}_s$.

The transition from State~3 to State~4 uses \textsc{Recv} to add a request to receive the message with tag $t$, following which the \textsc{WaitRecv} can fire to move from State~4 to State~5.
Analogously to \textsc{WaitSend}, \textsc{WaitRecv} can execute when a message with the given tag is in the message buffer, and it updates the variable to be written with the message payload, removing the tag from $\mathcal{B}_r$ in the process.
Although it is not shown in the figure, $x$ has been updated in $c_1$ to contain 5.

In contrast, the bottom row illustrates the case when an MPI send blocks until the matching receive is executed.
Specifically, after the \textsc{Send} rule is used to transition from State~6 to State~7, the process with rank~0 cannot proceed any further.
The \textsc{WaitSend} rule cannot fire because the message payload buffer $\mathcal{B}_m$ is empty and also the \textsc{TransferNoWait} rule cannot fire because the process is already blocked on a wait.
The only rule that can unblock process with rank~0 is, therefore, \textsf{TransferOnWait}, which can fire whenever the receiving process with rank~1 reaches the corresponding wait call.
This is illustrated with a transition from State~7 to State~8, where the process with rank~1 first gets to the \textbf{wait} call via the \textsc{SeqStep} rule and then initiates the transfer via the \textsf{TransferOnWait} rule.
Subsequently, \textsc{WaitRecv} transitions the system from State~8 to State~9, and \textsc{WaitSend} completes the communication by advancing to State~5.

Note that in State~5, the final state in Figure~\ref{fig:buffers}, the message buffer $\mathcal{B}_m$ still contains the message tagged $t$.
Once the communication associated with the tag has been completed, it may be removed from $\mathcal{B}_m$ at any point according to the \textsc{FreeBuffer} rule.
The rule mimics the expected behavior of MPI runtime in that any memory used internally by the MPI should eventually be freed.
The exact timing is implementation-dependent and irrelevant to our proof.

The last rule of our operational semantics, \textsc{Barrier}, specifies that all processes must execute a barrier simultaneously.
Note that we adopt a slight abuse of notation by allowing the \textsc{Barrier} rule to apply to the final barrier, even though the rule formally presupposes the existence of a subsequent command. 

\subsection{Axioms}
\label{sec:proof:axioms}

Our core calculus represents programs that a DafnyMPI user could write.
For Dafny to verify these programs, the user must prove a number of properties.
These include properties Dafny enforces by default, such as loop termination, absence of division by zero, and absence of references to undeclared variables.
Several other properties are MPI-specific and are expressed as preconditions on the relevant DafnyMPI methods.
We refer to these properties as \emph{axioms}, since in developing our proof, we can assume they hold for any execution as they are discharged by the Dafny verifier.
Each axiom other than \textsc{AtStart} is implicitly parameterized by a particular program state: for any state reachable from the initial state, if the axiom's premises hold, then so do its conclusions.

It is important to note that Dafny only models the behavior of a single (arbitrary) process and can only inspect the part of the program state that the process has access to. 
As a result, none of the axioms mentions the message buffer $\mathcal B_m$, whose contents is internal to the MPI runtime.
Similarly, the axioms only check whether a given tag is present in the send or receive buffers when the process to which the axiom applies is the intended sender or receiver.
Additionally, because the Dafny verifier assumes all method calls terminate, it cannot itself account for the possibility of deadlock.
DafnyMPI builds on top of this foundation to establish deadlock freedom explicitly.
At the same time, Dafny's non-blocking assumption guarantees that the axioms hold in any state reachable under that assumption.
This allows us to reason about a program state by considering what must hold in future states that are guaranteed to be reached.
In particular, because Dafny proves that the final barrier count for any process is equal to $\textsc{BarrierCount}(N)$, we can show that the barrier count never exceeds this value, which is one of the lemmas introduced in Section~\ref{sec:proof:proof}.

The axioms should also be understood in conjunction with the reduction rules presented above. 
For example, we say that the \textsc{AtWait} axiom in Table~\ref{fig:axioms} guarantees that tags increase monotonically but this is only true when we consider the axiom in conjunction with the \textsc{WaitRecv} and \textsc{WaitSend} reduction rules in Figure~\ref{fig:reductions}.
Specifically, the two reduction rules always set the new tag $t_{i,k+1}$ to match the argument of the \textbf{wait} call, which, in turn, is guaranteed to be greater than the previous tag $t_{i,k}$ by the \textsc{AtWait} axiom.

\begin{table}[tb]
\caption{Axioms of the core calculus (implicitly parameterized by program state, except for \textsc{AtStart}).
}
\label{fig:axioms}
\small
\begin{tabular}{|@{\hspace{0.4em}}c@{\hspace{0.4em}}|@{\hspace{0.4em}}c@{\hspace{0.4em}}|@{\hspace{0.4em}}c@{\hspace{0.4em}}|} \hline
Axiom & Premises & Conclusions \\ \hline

\textsc{AtStart} & & $S_0 = (\bigcup_{0\leq i < N}\langle c;~\textbf{barrier}, \{\texttt{rank}\mapsto i, \texttt{size}\mapsto N\}, -1, 0\rangle, \emptyset, \emptyset,\emptyset)$ \\ \hline

\multirow{2}{*}{\textsc{AtSend}}
& $c_{i,k} = \textbf{isend}~e~x\textbf{; }c$
& \multirow{2}{*}{$i = \textsc{Sender}(v,N) \quad \langle x, \sigma_{i,k} \rangle \Downarrow \textsc{Message}(v) \quad v \notin \mathcal{B}_{s,k}$} \\
& $\langle e, \sigma_{i,k} \rangle \Downarrow v$ & \\ \hline

\multirow{2}{*}{\textsc{AtRecv}} 
& $c_{i,k} = \textbf{irecv}~e~x\textbf{; }c$ 
& \multirow{2}{*}{$i = \textsc{Receiver}(v,N)\quad v \notin \mathcal{B}_{r,k}$} \\
& $\langle e, \sigma_{i,k} \rangle \Downarrow v$ 
& \\ \hline

\multirow{4}{*}{\textsc{AtWait}} 
& & $v > t_{i,k}$ \\
& $c_{i,k} = \textbf{wait}~e\textbf{; }c$
& $\forall v'.~ t_{i,k} < v' < v \Rightarrow (\textsc{Sender}(v',N) \neq i \land \textsc{Receiver}(v',N) \neq i)$ \\
& $\langle e, \sigma_{i,k} \rangle \Downarrow v$ 
& $\textsc{Barrier}(b_{i,k},N) \leq v < \textsc{Barrier}(b_{i,k} + 1,N)$ \\
& & $(i = \textsc{Sender}(v,N) \land v \in \mathcal{B}_{s,k}) \lor (i = \textsc{Receiver}(v,N) \land v \in \mathcal{B}_{r,k})$ \\ \hline

\multirow{4}{*}{\textsc{AtBarrier}} 
& & $v > t_{i,k}$ \\
& $c_{i,k} = \textbf{barrier}\textbf{; }c$
& $\forall v'.~ t_{i,k} < v' < v \Rightarrow (\textsc{Sender}(v',N) \neq i \land \textsc{Receiver}(v',N) \neq i)$ \\
& $v = \textsc{Barrier}(b_{i,k}+1,N)$ 
& $\forall v.~ \textsc{Sender}(v,N) = i \Rightarrow v \notin \mathcal{B}_{s,k}$ \\
& & $\forall v.~ \textsc{Receiver}(v,N) = i \Rightarrow v \notin \mathcal{B}_{r,k}$ \\ \hline

\textsc{AtEnd} 
& $c_{i,k} = \textbf{skip}$ 
& $b_{i,k} = \textsc{BarrierCount}(N)$ \\ \hline

\multirow{3}{*}{\textsc{AtSet}} 
& & $\forall v.~ \textsc{Sender}(v,N) = i \Rightarrow (v \notin \mathcal{B}_{s,k} \lor x \neq \mathcal{B}_{s,k}(v))$ \\
& $c_{i,k} = \textbf{set}~x~e\textbf{; }c$
& $\forall v.~ \textsc{Receiver}(v,N) = i \Rightarrow (v \notin \mathcal{B}_{r,k} \lor x \neq \mathcal{B}_{r,k}(v))$ \\
& & $x \neq \texttt{rank}$ \\ \hline

\multirow{2}{*}{\textsc{AtRead}} 
& $\textsc{Receiver}(v,N) = i$ 
& \multirow{2}{*}{$x \notin \text{direct subexpressions of } c_{i,k}$} \\
& $\mathcal{B}_{r,k}(v) = x$ 
& \\ \hline
\end{tabular}

\end{table}

Table~\ref{fig:axioms} presents the axioms of our language. 
The \textsc{AtStart} axiom describes the initial state of the program as outlined above in Section~\ref{sec:proof:core}.
\textsc{AtSend} and \textsc{AtRead} state that a message with a given tag can only be sent or received by its designated \textsc{Sender} and \textsc{Receiver}.
Both axioms also require that the corresponding buffers are empty to disallow repeatedly sending or receiving the same message before calling \textbf{wait}. 
Additionally, \textsc{AtSend} guarantees that the message payload is as specified by the \textsc{Message} function.

The \textsc{AtWait} axiom enforces the tag ordering property. 
First, it guarantees that message tags used with \textbf{wait} increase monotonically for each process.
Second, it enforces that no \textbf{wait} call is skipped---if the process is the designated receiver or sender of a message, it must receive or send this message before moving on to messages with higher tags.
Third, the axiom ensures that no barriers are skipped between successive \textbf{wait} calls.
Finally, it ensures that \textbf{wait} is only called after the corresponding \textbf{isend} or \textbf{irecv}.

\textsc{AtBarrier} requires that processes call barriers in order and wait on all send and receive operations that should precede them.
Similar to the \textsc{Barrier} rule above, the \textsc{AtBarrier} axiom should be understood to apply to the final barrier as well.
The axiom also guarantees that the send and receive buffers are empty at the barrier, i.e., there may be no ongoing send or receive calls. 
Because the last command executed by every process is a barrier, this last requirement guarantees that every \textbf{irecv} and \textbf{isend} is eventually followed by a matching \textbf{wait}. 
At the end of each process's execution, the \textsc{AtEnd} axiom states that the process passed exactly \textsc{BarrierCount}$(N)$ barriers.

Finally, \textsc{AtSet} and \textsc{AtRead} prevent nondeterministic behavior by disallowing reads or writes to variables that are used in interprocess communication.
We explain in Section~\ref{sec:func} how these two axioms are enforced by Dafny in practice when discussing the safe use of read and write buffers.
\textsc{AtSet} also prevents changes to variable \texttt{rank}, which means we can rely on it always being equal to the process's actual rank in the operational semantics.

\subsection{Deadlock}

\begin{table}
\caption{Definition of deadlock and related terms.}
\label{tab:deadlock}
\small
\begin{tabular}{|l|l|} \hline
Term & Conditions \\ \hline
$\textsc{OnSend}((\langle \textbf{wait}~e\textbf{; }c',\sigma,t,b\rangle_i \cup \mathcal P,\mathcal B_r, \mathcal B_s, \mathcal B_m))$ & $\langle e,\sigma \rangle\Downarrow t'\quad i=\textsc{Sender}(t',N)\quad t' \notin \mathcal B_m$ \\\hline
$\textsc{OnRecv}((\langle \textbf{wait}~e\textbf{; }c',\sigma,t,b\rangle_i \cup \mathcal P,\mathcal B_r, \mathcal B_s, \mathcal B_m))$ & $\langle e,\sigma \rangle\Downarrow t'\quad i=\textsc{Receiver}(t',N)\quad t' \notin \mathcal B_m\quad t'\notin \mathcal B_s$ \\\hline
$\textsc{OnBarrier}(\langle \textbf{barrier}\textbf{; }c',\sigma,t,b\rangle)$ &  \\\hline
$\textsc{Terminated}(\langle \textbf{skip},\sigma,t,b\rangle)$ &  \\\hline
& $|\mathcal P_3| \neq N \quad |\mathcal P_4|\neq N$ \\
$\textsc{Deadlock}((\mathcal P_1 \cup \mathcal P_2 \cup \mathcal P_3 \cup \mathcal P_4,$ & $\textsc{OnRecv}(p,\mathcal B_r, \mathcal B_s, \mathcal B_m),~~\forall p \in \mathcal P_1$\\
$\color{white}\textsc{Deadlock}((\color{black}\mathcal B_r, \mathcal B_s, \mathcal B_m))$ & $\textsc{OnSend}(p,\mathcal B_r, \mathcal B_s, \mathcal B_m),~~\forall p \in \mathcal P_2$\\
& $\textsc{OnBarrier}(p),~~\forall p \in \mathcal P_3$ \\
& $\textsc{Terminated}(p),~~\forall p \in \mathcal P_4$ \\ \hline
\end{tabular}
\end{table}

Table~\ref{tab:deadlock} introduces the terminology used to formally define and reason about deadlock.
We begin by listing the conditions under which we consider a process to be blocked:
\begin{itemize}
	\item A process is blocked on a send if it is about to execute a \textbf{wait} for a message it sends, but the message has not yet been placed in the payload buffer. The sender remains blocked until the receiver executes the corresponding \textbf{wait}. (This condition is captured by the \textsc{OnSend} predicate in the table.)
	\item A process is blocked on a receive if it is about to execute a \textbf{wait} command for a message it receives, but neither the message payload buffer nor the send buffer contains the relevant tag. 
	This means the sender has not yet initiated the communication.
	(This condition is defined by the \textsc{OnRecv} predicate in the table.)
	\item A process may be blocked on a barrier if it is about to execute the \textbf{barrier} command, as defined by the \textsc{OnBarrier} predicate in the table.
	As long as not all other processes have reached the barrier, the process will remain blocked.
\end{itemize}

With these definitions in place, we can now formally define deadlock.
A deadlock is any state, other than the one in which all processes have terminated, where no reduction rule applies.
Note that we exclude the \textsc{FreeBuffer} rule from this condition because any state where only \textsc{FreeBuffer} applies eventually leads to a deadlock.
One can see that this definition of deadlock is equivalent to every process being blocked or terminated, except when all processes have terminated.
This is captured by the definition of \textsc{Deadlock} in Table~\ref{tab:deadlock}.
All processes in a deadlocked state must belong to one of four groups: those blocked while waiting for a receive or send operation to complete ($\mathcal P_1$ and $\mathcal P_2$ in the table), those blocked on a \textbf{barrier} ($\mathcal P_3$), and those that have terminated ($\mathcal P_4$). 
Additionally, it is not possible for all processes to belong to $\mathcal P_4$ (because then the program would have terminated) or $\mathcal P_3$ (because then the processes would pass the barrier).
Note that the \textsc{OnRecv} and \textsc{OnSend} predicates specifically rule out states where \textsc{WaitSend}, \textsc{WaitRecv}, \textsc{TransferOnWait}, or \textsc{TransferNoWait} could be applied.
With deadlock formally defined, we now proceed to prove deadlock freedom. 

\subsection{The Deadlock Freedom Proof}
\label{sec:proof:proof}

In this section, we provide an outline of the proof that our core calculus does not admit deadlocks.
The proof proceeds by selecting the process waiting on the message with the smallest tag and showing that its communication partner cannot be blocked, since the ordering on message tags would only permit blocking on a smaller tag.
The main theorem relies on a number of helper lemmas, which we list without proof in the interest of conciseness.
Both the lemmas and the theorem assume the axioms described in Section~\ref{sec:proof:axioms} hold, including the description of the initial program state.
The complete proof is included in the supplementary material.

\begin{lemma}[\textsc{NoTwoTagsTheSame}]

\textit{If two processes are about to execute \textup{\textbf{wait}} commands in a deadlocked state, the corresponding message tags must be different.}
\end{lemma}

\begin{lemma}[\textsc{NonEmptyBuffer}]

\textit{If a process is responsible for sending or receiving a message with tag $v\in \mathbb N$, and that process's max used tag value eventually reaches $v$, then at some point during execution, the message buffer must have contained the payload tagged v.}
\end{lemma}

\begin{lemma}[\textsc{BufferStaysNonEmpty}]

\textit{If the process is the intended receiver or sender of a message in the payload buffer but the message's tag is higher than that of any message the process has received or sent so far, then the message must remain in the buffer.}
\end{lemma}

\begin{lemma}[\textsc{NumberOfBarriers}]
    \textit{If the process has a barrier at the end of its command sequence, its internal barrier counter is strictly less than \textsc{BarrierCount}(N).}
\end{lemma}

\begin{theorem}[Deadlock Freedom]
There exists no reduction from the initial state $S_0$ to a deadlocked state.

\end{theorem}
\begin{proof}
Suppose, for the sake of contradiction, that there exists a reduction $S_0\longrightarrow \cdots \longrightarrow S_k$ and $\textsc{Deadlock}(S_k)$. 
By the definition of \textsc{Deadlock}, there must then exist some $\mathcal{P}_1$, $\mathcal{P}_2$, $\mathcal{P}_3$, $\mathcal{P}_4$, $\mathcal B_s$, $\mathcal B_r$, and $\mathcal B_m$ such that $S_k=(\mathcal{P}_1 \cup \mathcal{P}_2 \cup \mathcal{P}_3 \cup \mathcal{P}_4, \mathcal B_s, \mathcal B_r, \mathcal B_m)$ and $\mathcal{P}_1$, $\mathcal{P}_2$, $\mathcal{P}_3$, $\mathcal{P}_4$ have the properties described in Table~\ref{tab:deadlock}. 
There are two cases, either $\mathcal{P}_1 \cup \mathcal{P}_2 = \emptyset$ or $\mathcal{P}_1 \cup \mathcal{P}_2 \neq \emptyset$.

\begin{itemize}
    \item \textbf{Case 1}. Suppose $\mathcal{P}_1 \cup \mathcal{P}_2 = \emptyset$. 
    This means that all processes that are still running are waiting on some barrier.
    By the definition of \textsc{Deadlock}, $|\mathcal P_3|\neq N$. 
    In other words, there must be at least one process that has terminated.
    Let that process be $p_m$ with rank $m$.
    Because this process has terminated, the \textsc{AtEnd} axiom applies, meaning that $b_{m,k}=\textsc{BarrierCount}(N)$.  
    Note that all processes must always share the same $b$, since all processes pass barriers simultaneously and all processes start with $b=0$ according to the \textsc{AtStart} axiom. 
    However, we know that at least one process, let us call it $p_i\in \mathcal P_3$, has not terminated yet.
    By the \textsc{NumberOfBarriers} Lemma, we know that $b_{i,k}<\textsc{BarrierCount}(N)$. 
    However, we also know that $b_{i,k}=b_{m,k}=\textsc{BarrierCount}(N)$.
    Hence, we reached a contradiction.
    Therefore, $\mathcal{P}_1 \cup \mathcal{P}_2 \neq \emptyset$.
    
   \item \textbf{Case 2}. Suppose $\mathcal{P}_1 \cup \mathcal{P}_2 \neq \emptyset$.
    This means that at least one process must be blocked on a \textbf{wait} call for which it is the receiver or sender.
Let $m$ be the rank of the process that is waiting on a message with the smallest tag, namely that $c_{m,k}= \textbf{wait}~e_{m}\textbf{; }c'$ and $\langle e_{m}, \sigma_{m,k} \rangle \Downarrow v_{m}$, where $v_{m}$ is smaller than the respective value for any other process in $\mathcal{P}_1 \cup \mathcal{P}_2$.
Such a minimum exists because, according to \textsc{NoTwoTagsTheSame} lemma, no two processes may wait on the same tag. 
We will now assume that $p_m \in \mathcal{P}_1$ (the reasoning is symmetric when $p_m \in \mathcal{P}_2$.)
Because $p_m\in \mathcal{P}_1$, we know that $m=\textsc{Receiver}(v_m,N)$.

Now consider the process with rank $i=\textsc{Sender}(v_m,N)$.
We will now show that $t_{i,k} \geq v_{m}$. 
We know that $p_{i,k}\in\mathcal{P}_1\cup \mathcal{P}_2 \cup\mathcal{P}_3\cup\mathcal{P}_4$. So there are three subcases:

\begin{itemize}
    \item Subcase 1: $p_{i}\in\mathcal{P}_1\cup\mathcal{P}_2$. 
    Then $c_{i,k}=\textbf{wait}~e_i\textbf{; }~c'$ and $\langle e_{i}, \sigma_{i,k}\rangle\Downarrow v_i$. 
    By construction, we know that $v_{i}>v_{m}$. 
    By the \textsc{AtWait} axiom, we also know that $v_{i} > t_{i,k}$.
    By the same axiom, it must be the case that $t_{i,k} \geq v_{m}$ since there must be no tag between $v_{i}$ and $t_{i,k}$ for which $i$ is the sender. 
    \item Subcase 2: $p_{i}\in\mathcal{P}_3$. By the \textsc{AtBarrier} axiom applied to process $i$ in state $k$, it follows that \textsc{Barrier}($b_{i,k}+1,N) > t_{i,k}$. 
    By the \textsc{AtWait} axiom applied to process $m$ in state $k$, $v_m$ < \textsc{Barrier}($b_{m,k}+1, N$). 
    As discussed above, all processes have the same value of $b$ in any given state.
    Hence, we have $b_{i,k}=b_{m,k}$ and $v_m< \textsc{Barrier}(b_{i,k}+1, N)$.
    By the \textsc{AtBarrier} axiom applied to process $i$ in state $k$, there must be no tag between $t_{i,k}$ and $\textsc{Barrier}(b_{i,k}+1, N)$ for which $i$ is the sender, so, therefore, $t_{i,k}\geq v_m$.
    \item Subcase 3: $p_{i}\in\mathcal{P}_4$ meaning that $p_{i}$ has terminated. 
    Note that by the \textsc{AtStart} axiom, $p_m$ has a \textbf{barrier} command at the end of its command sequence.
    Since $p_m$ has not passed this barrier yet, the \textsc{NumberOfBarriers} lemma applies, and $b_{m,k} < \textsc{BarrierCount}(N)$.
    At the same time, the \textsc{AtEnd} axiom applies to the process with rank $i$, so $b_{i,k} = \textsc{BarrierCount}(N)$.
    However, this leads to a contradiction because, as discussed above, all processes have the same value of $b$ in any given state, so $b_{i,k}=b_{m,k}$.
    Therefore, Subcase 3 is not possible.
    \end{itemize}

We have now established that $t_{i,k} \geq v_{m}$. 
Consider also that by the \textsc{AtWait} axiom, $v_m > t_{m,k}$. Since $t_{i,k} \geq v_{m} > t_{m,k}$, we can invoke the \textsc{NonEmptyBuffer} lemma for process $i$ and the \textsc{BufferStaysNonEmpty} lemma for process $m$. 
This gives us that $v_{m}\in B_{m,k}$. 
However, this is a contradiction because the \textsc{Deadlock} predicate requires that $v_{m}\notin B_{m,k}$. 
Therefore, $\mathcal{P}_1 \cup \mathcal{P}_2= \emptyset$.
\end{itemize}

We have now shown that both cases lead to a contradiction.
Therefore, our initial assumption must be wrong and \textsc{Deadlock} is not reachable.
\end{proof}

The proof above connects the behavior described by the MPI standard---formalized as the operational semantics of the core calculus---with the specifications we write in Dafny---formalized as the axioms of the core calculus---to establish that the specifications prevent deadlock in programs that use DafnyMPI.
The only other instance of interprocess reasoning necessary to establish full program correctness concerns the contents of the messages, which we discuss next.

\section{Proving Functional Equivalence}
\label{sec:func}

Our approach guarantees that a program written using DafnyMPI will never deadlock, but it does not by itself guard against domain-specific errors. % that a programmer may make, e.g., by incorrectly computing the five-point stencil or making some other application-specific mistake.
At the same time, it is not unusual for software engineers to make such mistakes when writing concurrent code.
Therefore, to further ensure the correctness of DafnyMPI programs, we propose a method for establishing functional equivalence between a parallel program written using DafnyMPI and a sequential program written in plain Dafny.
We follow this method to prove functional correctness of all three benchmarks described in Section~\ref{sec:evaluation}.
Specifically, we prove that the zeroth process---by convention the one that aggregates the results of all processes at the end---produces the same result as the sequential version.
The method consists of the following key steps:
\begin{itemize}
	\item First, the programmer writes in Dafny a sequential version of the computation they wish to parallelize.
	Sequential code is typically much easier to understand, reducing the likelihood of domain-specific bugs at this stage.
	\item Next, the programmer writes a ghost functional specification of the program and proves that the sequential version returns the same result as the functional specification for all inputs. 
	The functional specification is typically more difficult to understand since all loops have to be replaced with recursion. 
	This proof is reasonably straightforward, since this is what Dafny is commonly used for (e.g., see~\cite{aws}).
	We discuss some proof engineering considerations, particularly how we deal with proof brittleness, in Section~\ref{sec:engineering}.
	\item Next, the programmer decides on the topology of the parallel version and writes the corresponding ghost specification of the communication as in the first snippet of Dafny code in Figure~\ref{fig:example}.
	This specification may call out to the previously developed functional specification to describe the payloads of all the messages.
	\item Finally, the programmer writes the parallel version of the program while ensuring that (i) the zeroth process returns the same result as the functional specification (therefore being equivalent to the sequential implementation by transitivity), and (ii) the program abides by topology specification (enforced by DafnyMPI). 
	This is typically the most difficult part of the process.
	In our proofs, we follow the top-down approach by sketching the main program loop first and assuming lower-level methods behave as expected, then proving the correctness of these lower-level methods.
\end{itemize} 

To reason about the return value of a process in an MPI program, and, consequently, to establish functional equivalence, we must reason about the values returned by the individual MPI operations.
DafnyMPI allows such reasoning by establishing two properties.
First, DafnyMPI uses rely-guarantee reasoning to prove that the messages received by any process are exactly as specified by the user.
Second, DafnyMPI enforces safe use of buffers, to prevent situations when two receive operations write simultaneously to the same buffer. 
Together, these properties guarantee that the content of the write buffer is exactly as specified by the user at the time when the corresponding read operation terminates.
We now discuss these two properties in more detail.

\paragraph{Rely-Guarantee Reasoning}
For Dafny to verify properties about the values computed by each process, it must know the contents of the messages these processes receive.
It is clear that the contents will always be exactly as specified by the \textsc{Message} function (introduced in Section~\ref{sec:proof:core}).
Specifically, the \textsc{AtSend} axiom guarantees that the variable from which the message originates abides by the \textsc{Message} function and the \textsc{AtSet} and \textsc{AtRead} axioms prevent the process or another MPI operation from writing to this variable until the communication completes.
Because the \textsc{Message} function is a pure function of the message tag and the number of processes, any two processes will always agree on the contents of the message with the same tag.
More formally, this is an example of rely-guarantee reasoning.
Dafny can rely on each message received by a given process to abide by the user-provided specification because it guarantees that every message ever sent also abides by the specification.
In the Dafny code, this is captured by a postcondition on each \textbf{wait} call, which describes the message being received. 

\paragraph{Safe Use of Read and Write Buffers}

Our core calculus prevents writes to variables used in MPI communication (via the \textsc{AtSet} axiom) and prevents any use of variables that an MPI operation may already be writing to (via the \textsc{AtRead} axiom).
These restrictions are necessary because, in the presence of non-blocking operations, a process can execute concurrently with the MPI operations it initiated.
To enforce these guarantees in Dafny, we must additionally account for aliasing and the fact that real MPI implementations perform communication through buffers, i.e., arrays of data shared by reference.
In practice, MPI libraries may read from or write to the designated portion of a buffer at any time until the corresponding operation completes.
Send commands can, and often do, use overlapping buffers, but each receive command must have exclusive access to the relevant portion of the buffer.
Additionally, a process cannot write to a buffer currently in use by any MPI operation or read from any buffer that is being written to by an MPI call.

Enforcing the correct use of buffers is possible in Dafny because Dafny abstracts away direct memory access, meaning the only way to access a buffer is by calling an instance method on the relevant object.
To support buffers, we introduce two custom classes for 1D and 2D arrays of floating-point numbers.
Both classes implement a shared interface (a.k.a. \emph{trait} in Dafny) for modeling memory contiguity and per-element read and write locks as part of the ghost state. 
For each array element, the object's ghost state tracks whether it is currently being written to by a receive operation or is being read by (possibly multiple) send operations. 
Write locks act as flags: when a portion of an array is locked for writing, that region cannot be read or written until the lock is released.
Read locks, in contrast, act as counters: multiple MPI operations may read the same elements concurrently, but writes are disallowed as long as at least one operation is still reading them.
Each send or receive command takes as input a reference to the array used as a buffer, together with the starting index and the size of the message.
The preconditions of the corresponding MPI call ensure the specified array elements are available for reading or writing.
If the MPI call is non-blocking, the postconditions then mark the relevant elements of the array as locked until the operation completes.
For 2D arrays, we model C-style contiguity (which is guaranteed at  runtime by Python's Numpy library as discussed in~\ref{sec:engineering}), meaning that a row or series of rows are contiguous in memory but columns are not, unless the array only has one column.

Together, the rely-guarantee reasoning about message payloads and the read/write-lock mechanism controlling buffer access guarantee that the content of the received messages is exactly as specified by the user, thereby enabling proofs of functional equivalence.

\section{Implementation}
\label{sec:engineering}

DafnyMPI consists of two parts.
The first part is the specifications of MPI primitives, which are instance methods of the \texttt{MPI.World} class.
These specifications, along with a few supplementary lemmas and functions, comprise 379 lines of Dafny code and form the trust base of the library.
This file does not get compiled to Python directly; instead, each specified MPI primitive is implemented as a Python wrapper function (about 60 lines of code total), which we provide alongside the library.

The second part of DafnyMPI defines two classes for manipulating 1D and 2D arrays of real numbers, with the added ability to track  which elements are currently in use by MPI (see above).
For maximum portability, all array operations such as rolling, slicing, scalar and per-element multiplication and addition, etc. are implemented in Dafny and rely only on calls to the constructor and three primitive \code{Get}, \code{Set}, and  \code{Size} methods.
This part of the library totals 1772 lines of Dafny code, which can be directly compiled to Python (or any other target language supported by Dafny) with the exception of the constructor and the three primitive methods mentioned above.
The latter are all implemented in Python using the Numpy library, and the code is provided alongside the DafnyMPI library.
Numpy's default array constructor also guarantees C-style contiguity of array elements in memory, which is a property our Dafny code assumes on constructing an array and maintains throughout. 
In practice, all array operations could be delegated to Numpy at runtime to optimize for performance, but we opted to minimize the trust base instead.

Next, we discuss some of the relevant implementation challenges.

\paragraph{Dealing with Tag Overflow.}

The proof of deadlock freedom in Section~\ref{sec:proof:proof} assumes that MPI tags are unbounded.
In reality, different MPI implementations impose different limits on tag size, e.g., the Intel MPI library uses 20 bits to represent tags~\cite{intelmpi}.
The only requirement imposed by the MPI Standard itself is that the upper bound on the tag must be no less than 32767~\cite{mpi-spec}.

To ensure tags are always valid, DafnyMPI adds a proof obligation to require that the \textsc{Barrier} function we introduce in Section~\ref{sec:proof:core} never increases by more than $32767$.
This ensures that at most $32767$ tags are used between any two barriers.
Therefore, even if tag overflow occurs, tags remain unique between barriers, allowing processes to deterministically identify messages. 

\paragraph{Supporting Blocking Send and Receive.}

Our core calculus models only non-blocking send and receive operations.
As mentioned earlier, we can easily encode the blocking variants as an \textbf{isend} or \textbf{irecv} followed immediately with a \textbf{wait}.
The DafnyMPI library supports blocking send and receive by expressing them in this way for verification purposes while using the corresponding \texttt{MPI\_SEND} and \texttt{MPI\_RECV} commands at runtime.
Support for blocking sends and receives simplifies the code by removing the need to manage \texttt{MPI\_REQUEST} objects when non-blocking communication provides no additional benefit.

\paragraph{Supporting Collective Operations}

The MPI standard describes a number of one-to-many, many-to-one, and many-to-many collective operations such as \texttt{MPI\_SCATTER}, \texttt{MPI\_GATHER}, and \\\texttt{MPI\_ALLREDUCE}. 
The core calculus does not include these features, but DafnyMPI does support the latter two operations by treating them as decorated barriers.
Specifically, DafnyMPI imposes a restriction on collective operations requiring that they are used only when there are no pending send or receive operations running in the background.
At runtime, DafnyMPI explicitly puts a barrier before any such collective operation.
Thus, the addition of these new collective operations cannot introduce deadlocks.
Any information necessary to reason about the output of a given operation is provided by the user when the MPI is initialized, similar to how the one-to-one messages are specified. 

For example, to describe the \code{Gather} operation on line~\ref{line:gather} of our running example, the user must specify that the final \code{nt}th barrier is a \code{Gather} operation and provide the exact functional specification of the information that will be collected (not shown in Figure~\ref{fig:example} in the interest of conciseness). 
At the call to \code{Gather}, in addition to any conditions that are required for a regular barrier, the user must also prove that the data sent by each process matches the respective portion of the specified result and that the relevant buffers can be read from and written to. 
The specification of \code{Gather} then guarantees that the data received by the root process will be exactly as specified above.

\paragraph{Termination}
Of all the correctness properties we list in Section~\ref{ref:overview:dafny}, proving termination requires the least amount of work. 
Once deadlock freedom is established and the verifier successfully rules out runtime errors such as out-of-bounds array access,  loops and recursions remain as the only obstacles to termination.
We establish that loops and recursion terminate as is standard in Dafny, by specifying some positive quantity that decreases on each iteration and is bounded below by zero.
For the outer for loop in our running example in Figure~\ref{fig:example}, this quantity is \code{nt-n}, the number of iterations that remain.

\paragraph{Proof engineering considerations}
\label{sec:engineering:proof}

A common issue that arises in automated verification is \emph{brittleness}: the more information one gives to the verifier, the more likely it is that the verification query times out. 
To address this problem, we follow the approach previously introduced by Dafny  engineers~\cite{opaque1, opaque2} in making most functions---particularly those involving quantifiers---\emph{opaque}, hiding their definitions from the verifier unless explicitly revealed.
We also specify the behavior of some such functions using specialized lemmas instead of using postconditions, which similarly gives the user more control over which facts the verifier has access to at any given point.
We recommend that the clients of DafnyMPI follow the same approach to avoid brittleness in their proofs.

\section{Evaluation}
\label{sec:evaluation}

To evaluate whether DafnyMPI meets its design goals, we implemented and verified three common numerical PDE solvers.
We evaluated the proof engineering effort, in terms of the amount of auxiliary proof-related code, and runtime performance, in terms of performance of the parallel implementation.
Our results suggest that DafnyMPI can offer significant runtime improvements, with speedups of up to five times, while keeping the added proof burden manageable, even as the amount of proof-related code increases compared to sequential implementations.
We conducted all experiments on a MacOS machine with an M3 processor and 48 GiB of memory, and the Dafny solver was configured to use an unlimited number of virtual cores for verification.

\begin{figure}
\begin{subfigure}[t]{0.32\textwidth}
\includegraphics[width=\linewidth]{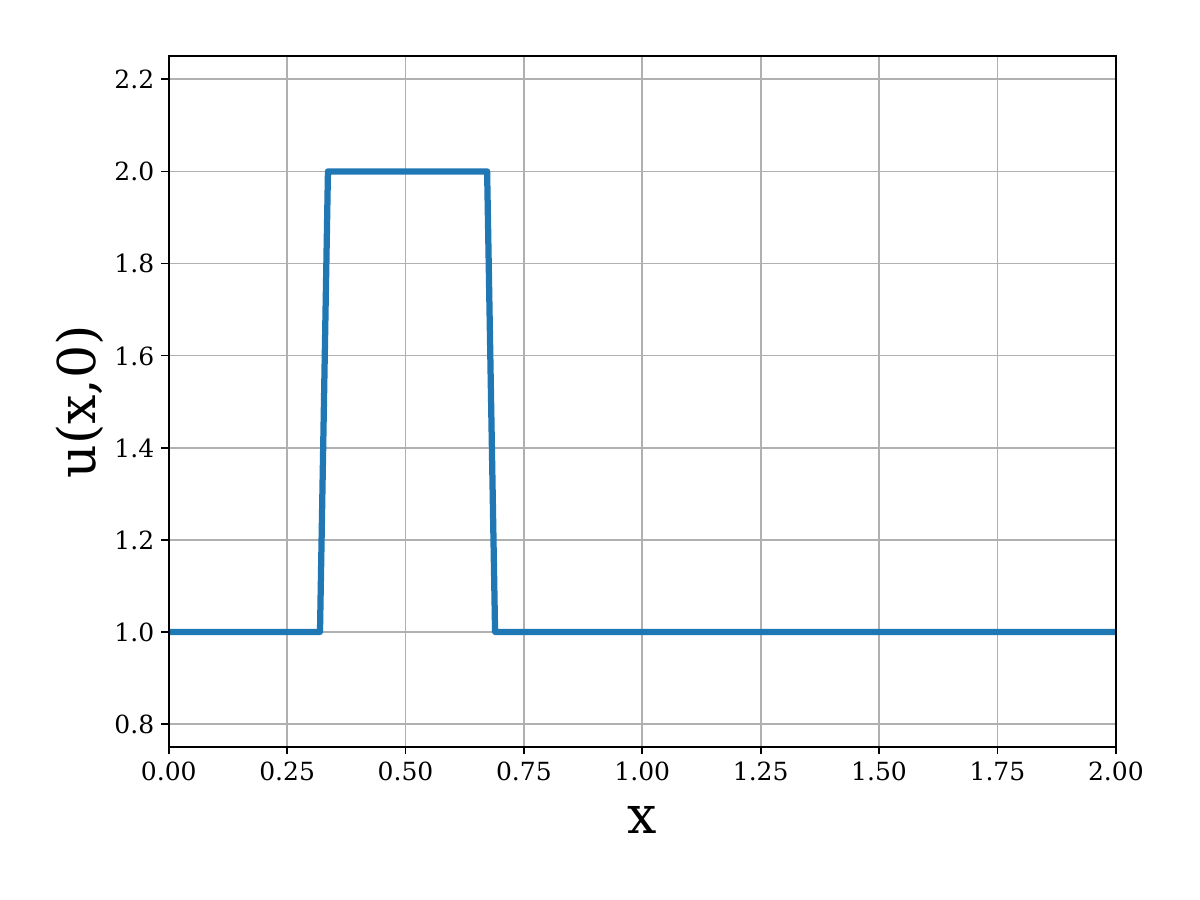}
\end{subfigure}
\begin{subfigure}[t]{0.32\textwidth}
\includegraphics[width=\linewidth]{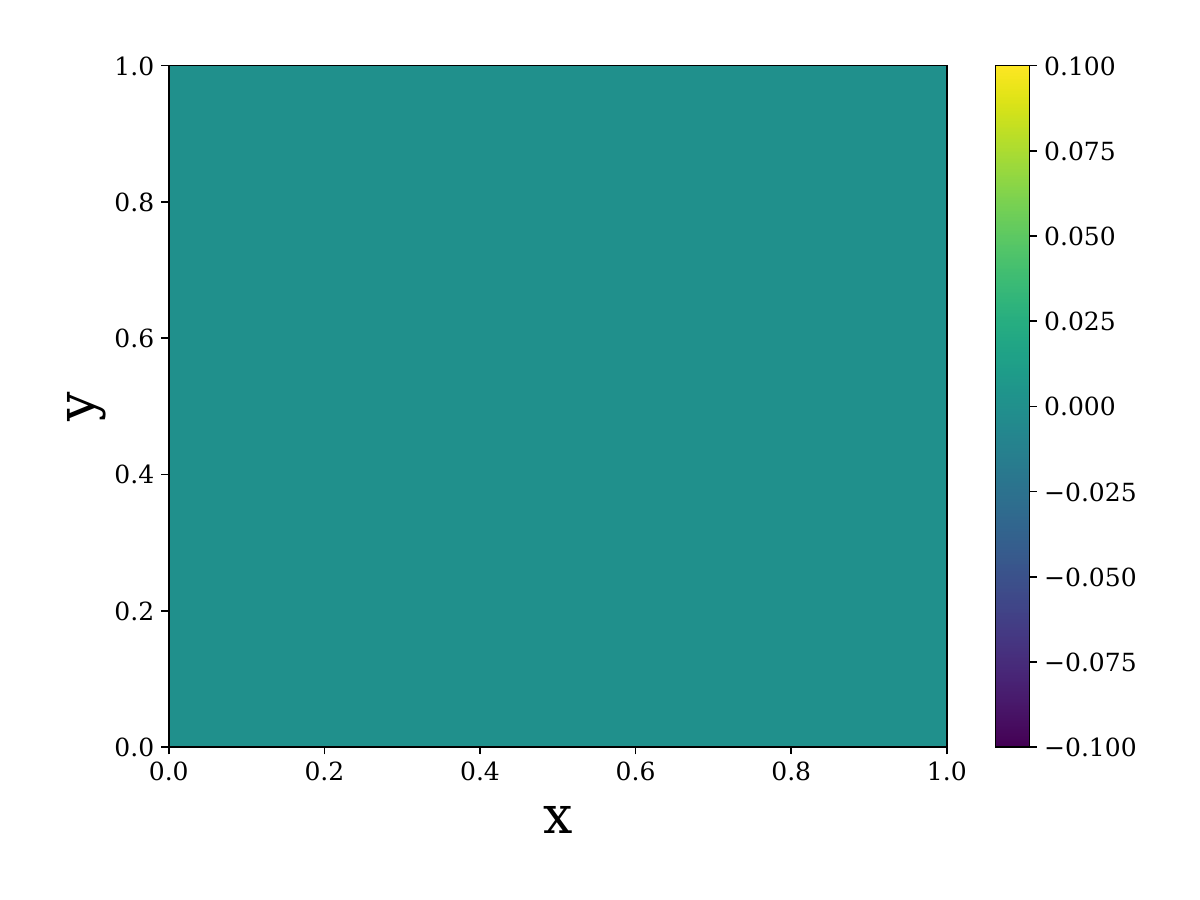}
\end{subfigure}
\begin{subfigure}[t]{0.32\textwidth}
\includegraphics[width=\linewidth]{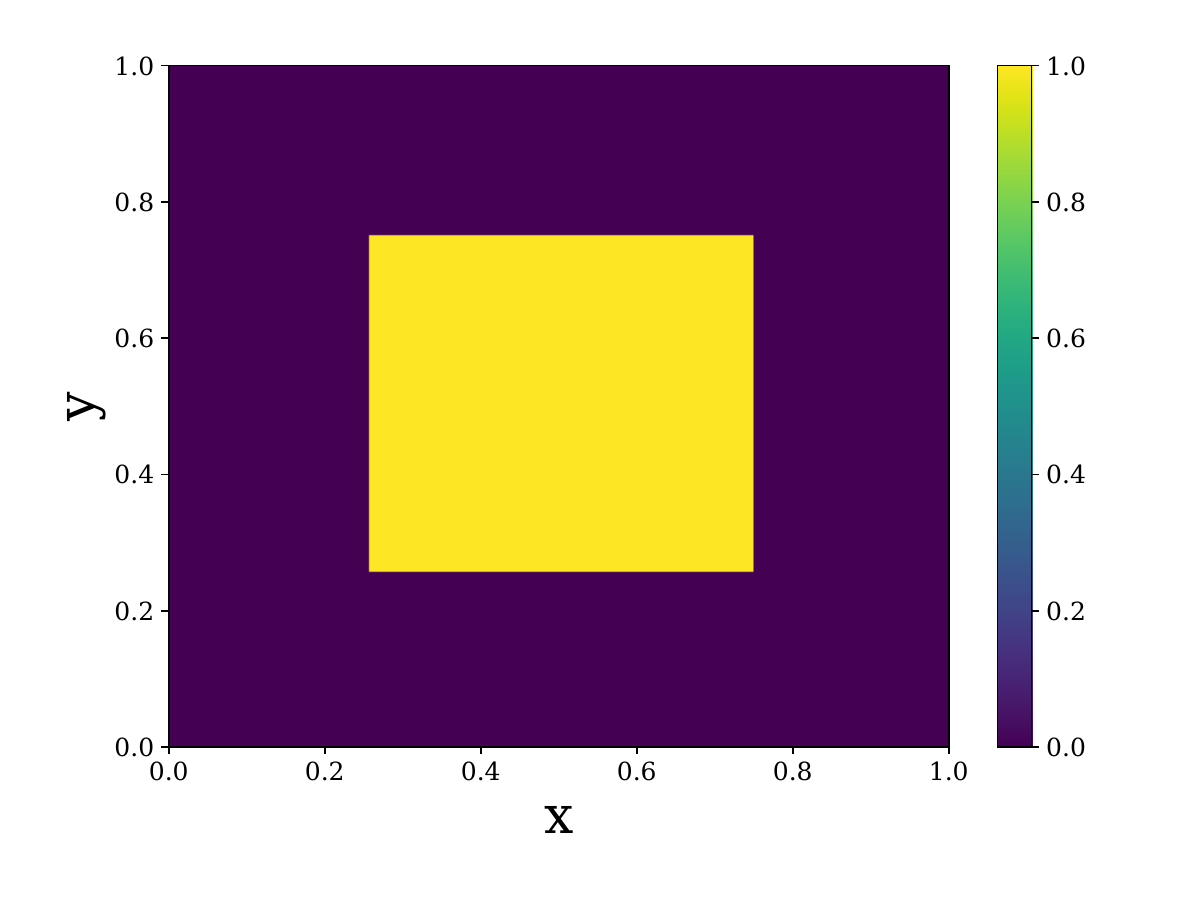}
\end{subfigure}
\begin{subfigure}[t]{0.32\textwidth}
\includegraphics[width=\linewidth]{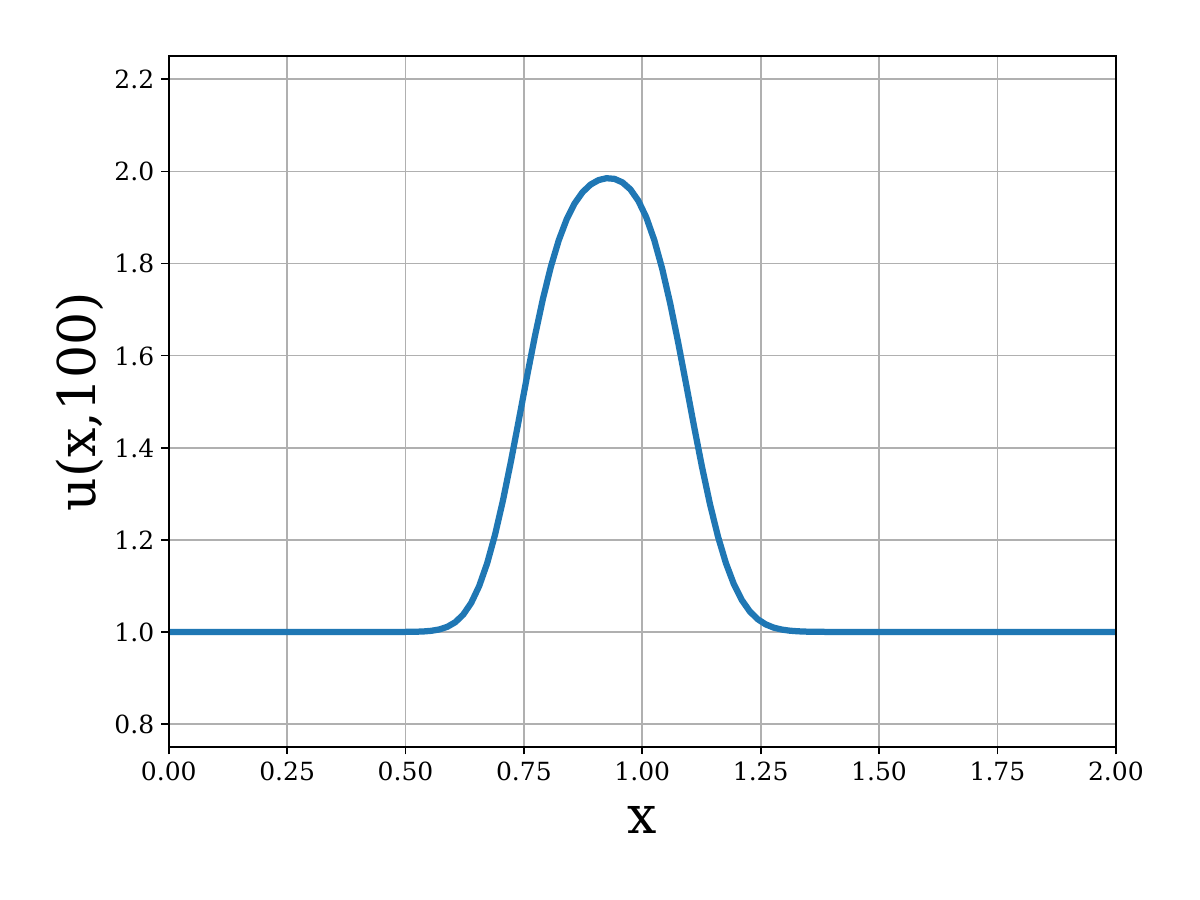}
\end{subfigure}
\begin{subfigure}[t]{0.32\textwidth}
\includegraphics[width=\linewidth]{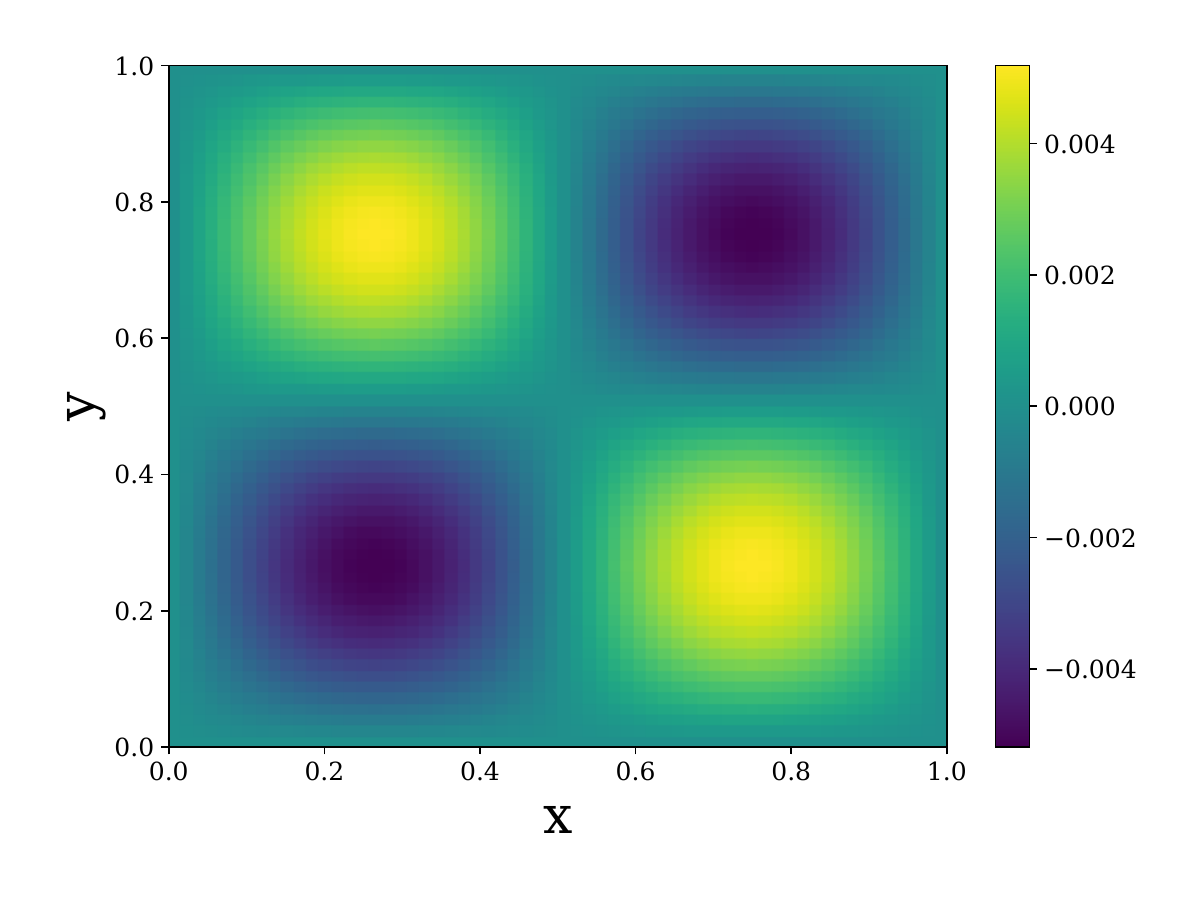}
\end{subfigure}
\begin{subfigure}[t]{0.32\textwidth}
\includegraphics[width=\linewidth]{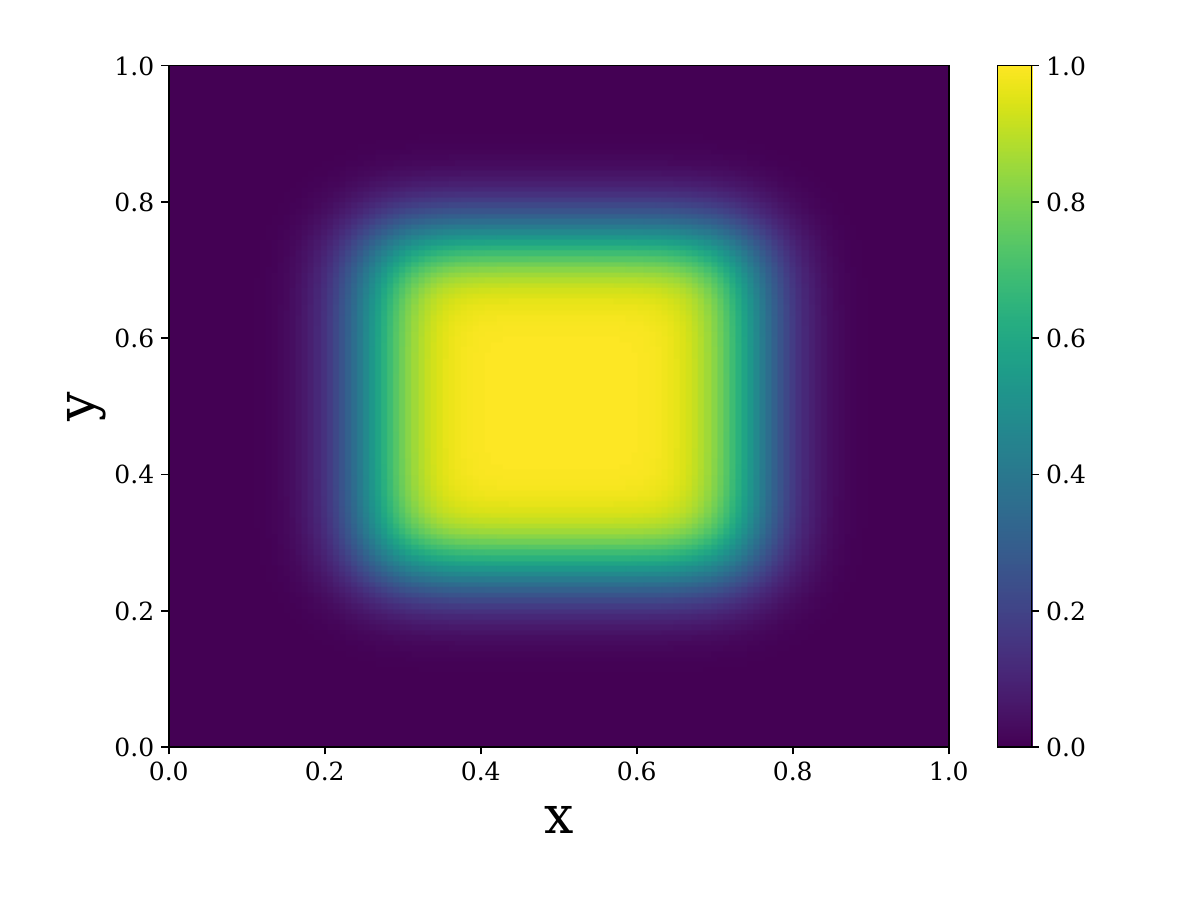}
\end{subfigure}
\caption{Benchmark PDE simulations. Top row: initial conditions; bottom row: solutions after 100 time steps as computed by DafnyMPI. From left to right: Linear Convection, Poisson, Heat Diffusion}
\label{fig:benchmarks}
\end{figure}

\subsection{Benchmarks} 

We used DafnyMPI to implement numerical solvers for three different PDEs, each of which uses a different communication pattern.
Figure~\ref{fig:benchmarks} illustrates the solver results after simulating each PDE for 100 time steps using DafnyMPI.
Figure~\ref{fig:comm} sketches the MPI communication pattern for each solver.
The benchmarks are as follows:
\begin{itemize}
	\item \textbf{Linear Convection.} As discussed in Section~\ref{sec:overview}, we implemented the upwind scheme for solving the 1D linear convection PDE.
	The equation models a wave propagating at constant speed, as shown in the left column of Figure~\ref{fig:benchmarks}.
	Notice that while the analytical solution to linear convection preserves the wave shape exactly, the numerical solution introduces diffusion and cannot fully maintain sharp edges in the initial condition.
	The top left schematic in Figure~\ref{fig:comm} shows that, in the MPI version, each process must send the rightmost value in its domain to its neighbor at each iteration.
	\item \textbf{Poisson.} We implemented Jacobi iteration for solving the 2D Poisson equation with Dirichlet boundary conditions.
	The Jacobi method updates the value at a point based on all four of its neighbors on each iteration.
	For the MPI implementation, we partition the domain into equal horizontal stripes. 
	Each process must send the top and bottom rows of its stripe to its neighbors and receive the two boundary rows in return.
	The top right schematic in Figure~\ref{fig:comm} illustrates this communication pattern (we use columns instead of rows in the figure due to space constraints.)
	The solution to the Poisson equation is a function whose second spatial derivative (Laplacian) is equal to another function, provided as input. 
	Over the course of multiple iterations, the solution can be expected to converge to an equilibrium, so the processes must communicate on each iteration to check whether convergence has been reached. 
	The latter is handled via a collective call to \texttt{MPI\_ALLREDUCE}.
	The middle column of Figure~\ref{fig:benchmarks} shows the solution to the Poisson equation after 100 iterations, where we set the initial conditions to a matrix of zeroes and the right-hand side of the equation to a function $f(x,y)=\sin(2\pi x)\sin(2\pi y)$, whose values on the discretized domain are precomputed using externally implemented trigonometric functions.
	\item \textbf{Heat Diffusion.} We implemented the fourth-order Runge-Kutta method (RK4) and applied it to the 2D heat diffusion equation. 
	The right column of Figure~\ref{fig:benchmarks} illustrates the behavior described by the equation using a heatmap. 
	The heat that is originally at the center of the plot gradually dissipates as the simulation proceeds. 
	The fourth-order Runge-Kutta method computes the value at a point based on a region around it that includes any points within the Manhattan distance of 4.
	For the MPI implementation, we again partition the domain into equal horizontal stripes.
	Unlike in Poisson, however, the processes must communicate 4 rows of data at a time, which introduces an additional layer of complexity to the problem.
	The bottom schematic in Figure~\ref{fig:comm} illustrates this process (as before, we use columns instead of rows.) 
	Note that the data the process sends to its two neighbors can be overlapping, and the sends can happen concurrently. 
	The receives, however, must use separate buffers.
	\end{itemize}

\subsection{Proof Engineering Effort}

To provide insight into the engineering effort involved, Table~\ref{tab:effort} reports the number of lines of code (LOC) for each of the three benchmarks described above.
While it is only an indirect measure, LOC roughly corresponds to the amount of effort that went into verifying each respective part of the code.
We separately list the number of lines in the specification (Spec), sequential implementation (Seq), and parallel MPI implementation (Par) for each benchmark. 
For the \textsc{Heat Diffusion} and \textsc{Poisson} benchmarks, we also list the size of the codebase that is shared between the sequential and the parallel implementations (Shared), which corresponds to several methods for implementing the Jacobi and RK4 techniques, respectively.
As can be expected, due to the increased complexity, the parallel version requires more lines of code to implement.

\paragraph{Ghost vs Executable Code.} Table~\ref{tab:effort} also reports the percentage of lines in each group of files that are \emph{ghost}.
 Ghost code is code that is not executable and whose only purpose is to support the proof.
As can be seen from the figure, the percentage of ghost code increases significantly for the parallel version, which reflects the increased complexity of the proof.
At the same time, the percentage of ghost code in the parallel implementations is comparable to that in some of the more complex subroutines of the sequential version.
For example, the 4th-order Runge-Kutta method (the Shared subcategory of the Heat Diffusion benchmark in Table~\ref{tab:effort}) requires a significant amount of auxiliary proof code to relate the functional specification to the rolling array operations used in the implementation.
As a result, the ghost code makes up 81\% of this shared component, which is close to the 89\% observed in the parallel implementation of the same benchmark.

\begin{table}[tb]
\centering{\setlength{\tabcolsep}{4pt}
\caption{Number of lines of code (LOC) and verification time as measures of proof complexity.}
\begin{tabular}{|l|r|r|r|r|r|r|r|r|r|r|r||r|}
\hline
\multirow{2}{*}{} & \multicolumn{3}{|c|}{Linear Convection} & \multicolumn{4}{|c|}{Poisson} & \multicolumn{4}{|c||}{Heat Diffusion} & \multirow{2}{*}{Total} \\
\cline{2-12}
 & Spec & Seq & Par & Spec & Shared & Seq & Par & Spec & Shared & Seq & Par &  \\ \hline\hline
LOC & 209 & 106 & 379 & 384 & 126 & 153 & 1626 & 328 & 406 & 137 & 1516 & 5370 \\ \hline
\% Ghost & 100 & 45 & 65 & 100 & 72 & 55 & 91 & 100 & 81 & 57 & 89 & 87 \\ \hline
Time to & \multirow{2}{*}{3} & \multirow{2}{*}{3} & \multirow{2}{*}{25} & \multirow{2}{*}{4} & \multirow{2}{*}{6} & \multirow{2}{*}{4} & \multirow{2}{*}{148} & \multirow{2}{*}{3} & \multirow{2}{*}{11} & \multirow{2}{*}{4} & \multirow{2}{*}{154} & \multirow{2}{*}{371} \\
verify (s) & \rule{0pt}{2.2ex} & \rule{0pt}{2.2ex} & \rule{0pt}{2.2ex} & \rule{0pt}{2.2ex} & \rule{0pt}{2.2ex} & \rule{0pt}{2.2ex} & \rule{0pt}{2.2ex} & \rule{0pt}{2.2ex} & \rule{0pt}{2.2ex} & \rule{0pt}{2.2ex} & \rule{0pt}{2.2ex} & \rule{0pt}{2.2ex} \\ \hline
\end{tabular}
\label{tab:effort}
}
\end{table}

\begin{figure}
\begin{subfigure}[t]{0.32\textwidth}
\includegraphics[width=\linewidth]{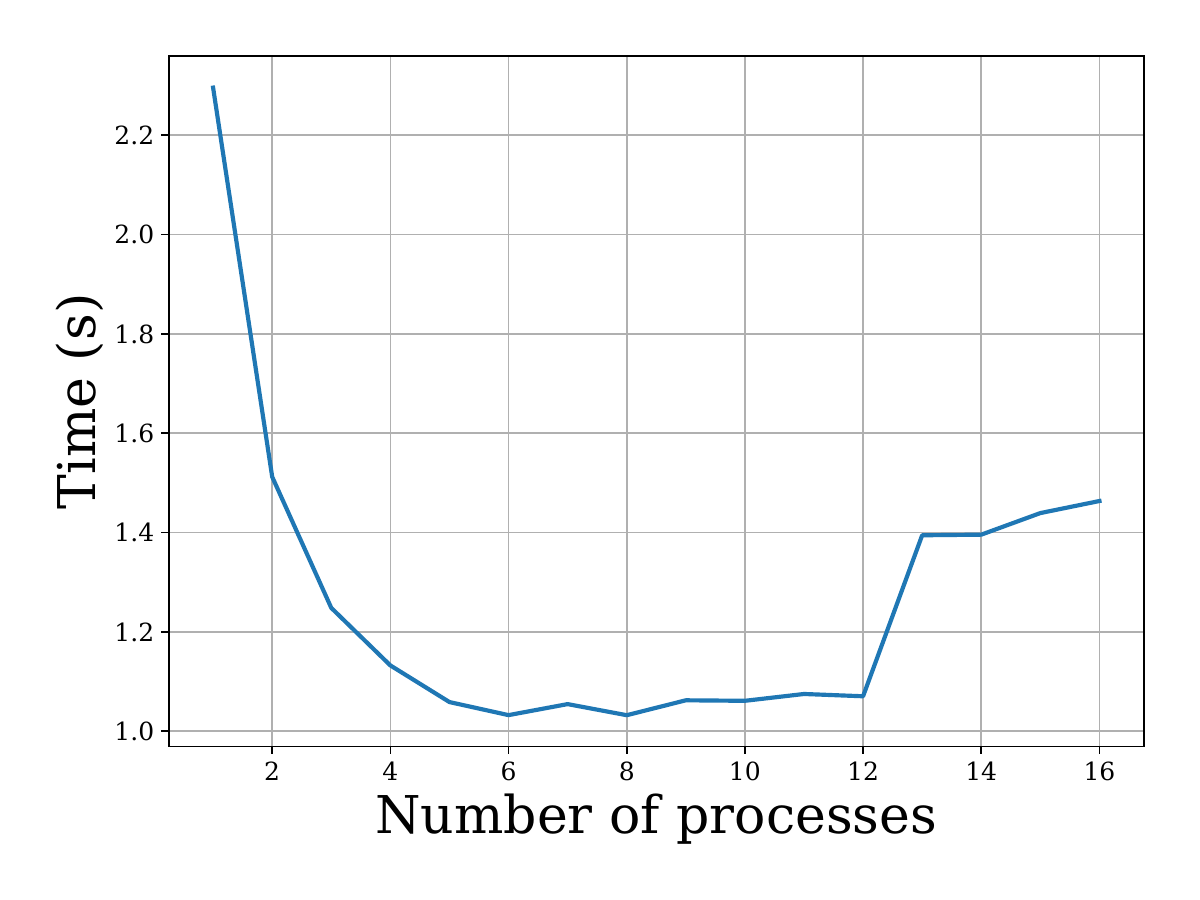}
\end{subfigure}
\begin{subfigure}[t]{0.32\textwidth}
\includegraphics[width=\linewidth]{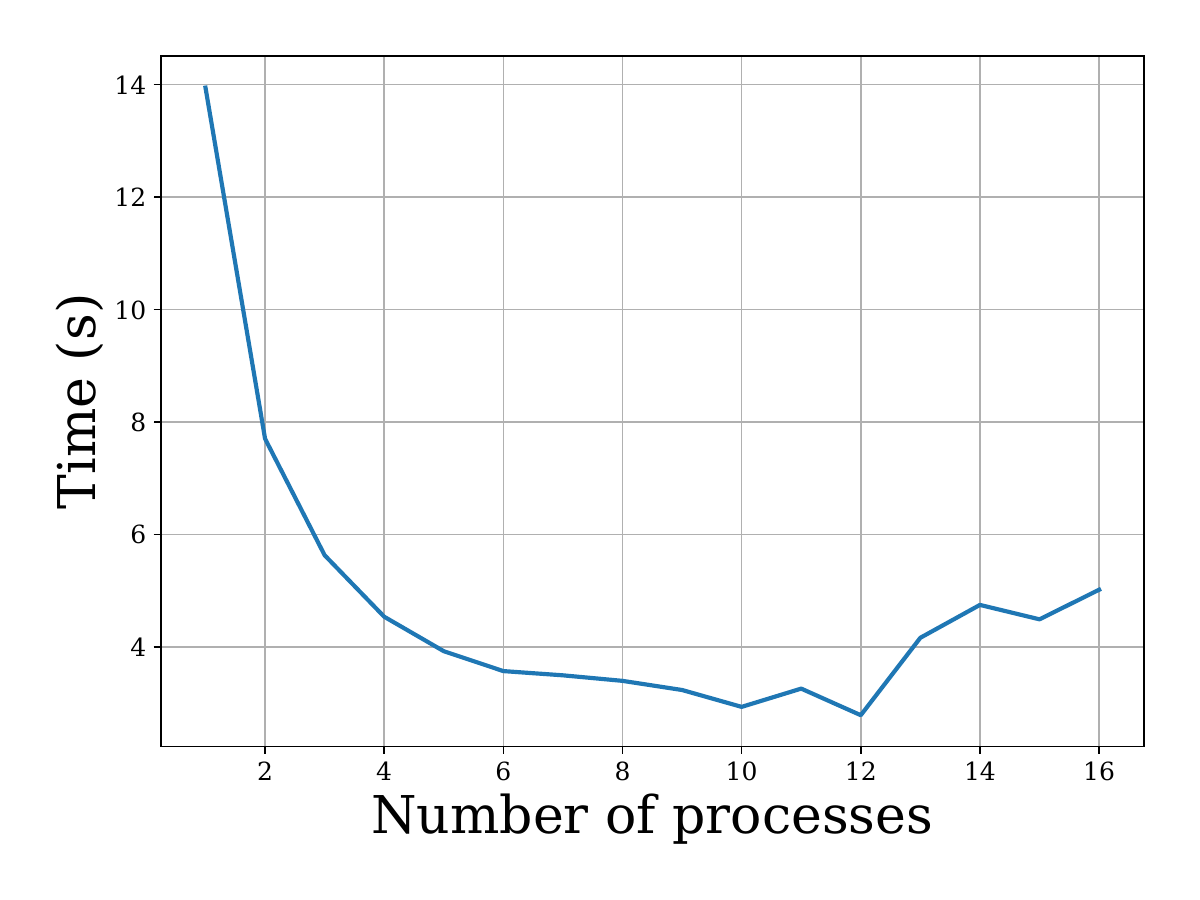}
\end{subfigure}
\begin{subfigure}[t]{0.32\textwidth}
\includegraphics[width=\linewidth]{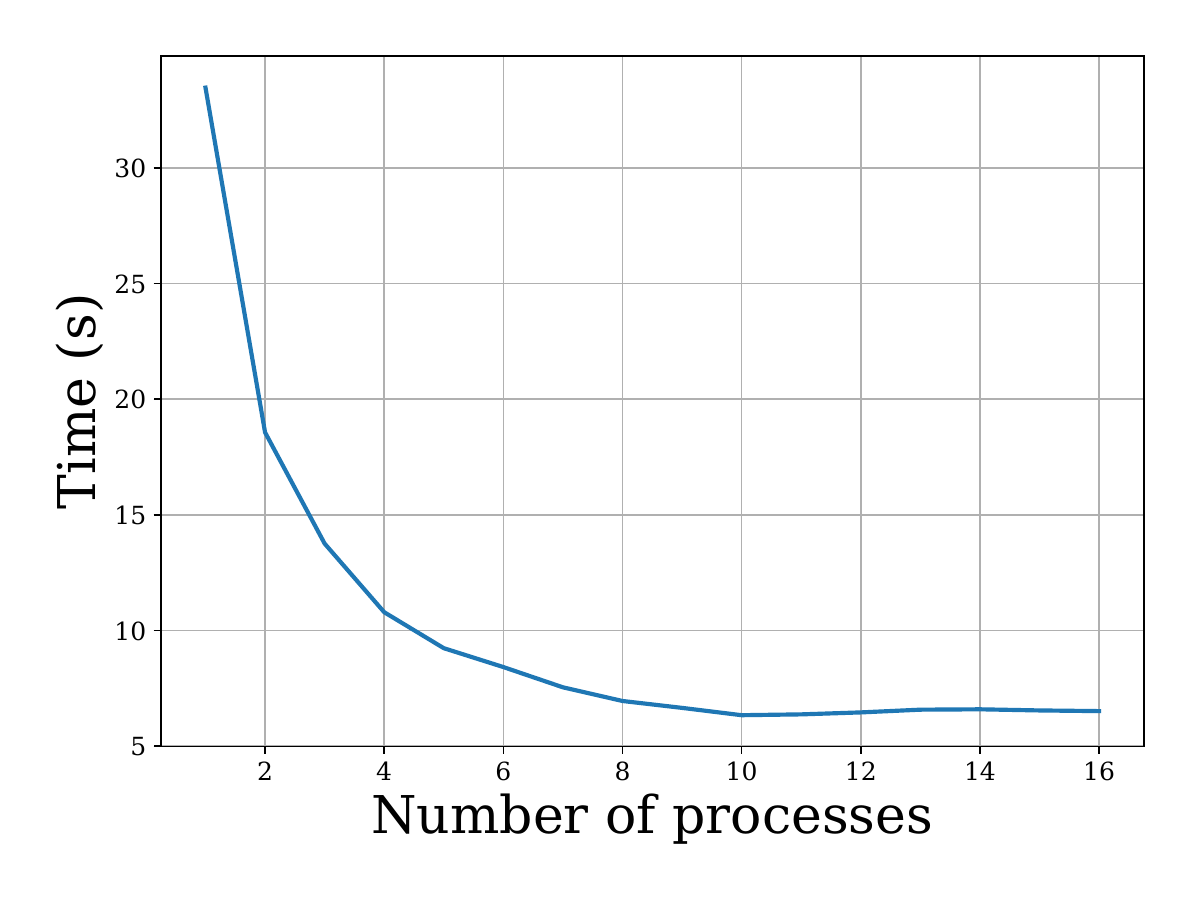}
\end{subfigure}
\caption{Running time as function of the process count. From left to right: Linear Convection, Poisson, Heat Diffusion.}
\label{fig:performance}
\end{figure}

\paragraph{Verification and Compilation Speed} It takes about 6 minutes to verify all of the benchmarks and under 20 seconds to compile the code to Python using Dafny's built-in compiler.
\textsc{Heat Diffusion} is the most expensive benchmark when it comes to verification time, as can be seen in Table~\ref{tab:effort}.
This is primarily due to the need to map the message payload between 2D and 1D representations multiple times during verification to ensure that the data being sent is contiguous in memory.

\subsection{Runtime Performance}

Given the considerable effort that goes into proving an MPI program correct, we want to be certain that the programs we write provide the desired performance boost when compared with corresponding sequential implementations.
Thus, we transpiled all three benchmarks to Python using the Dafny compiler and ran them with different numbers of processes. 
\footnote{
For simplicity of reasoning about the program's result, our current proofs assume that the domain size is divisible by the number of processes.
This assumption does not affect the communication pattern or the proof of deadlock freedom.
When the domain size is not divisible by the number of processes, we round it up to the nearest divisible value before measuring running time.
In our benchmarks, all domain sizes are divisible by 16, so experiments with 1, 2, 4, 8, and 16 processes are unaffected by this adjustment.
}
Figure~\ref{fig:performance} shows the results of this experiment.
All time measurements are the median of three runs. 
We use the sequential implementation for the one-process run.
As expected, increasing the number of processes reduces the running time, achieving up to a fivefold speedup over the sequential version in the Heat Diffusion benchmark.
However, as is typical, increasing the number of processes past a certain threshold leads to diminishing returns, and, for  Linear Convection and Poisson, the running time even goes up once we assign 12 or more processes to the MPI program. 
We hypothesize that the running time increases because we ran the experiments on a machine with 16 cores, so giving 12 cores to the program might have interfered with the other processes running in the background, thereby increasing the total running time. 
This effect is most noticeable in the Linear Convection benchmark because this is the most lightweight of the benchmarks in terms of running time.

\section{Discussion and Future Work}
\label{sec:discussion}

\paragraph{Comparison with a Proof Assistants}

An initial prototype of the DafnyMPI system was partially formalized in the Rocq proof assistant. 
To achieve this, we developed a domain-specific language similar to the core
calculus in Section~\ref{sec:proof:core}, alongside specifications for an interpreter. 
While we still believe that this approach is feasible, during the verification of the linear convection benchmark in Rocq we found that the proof engineering burden could be lessened considerably using the constructs built-in to Dafny. 
The benefit is most apparent when comparing the length of the partial proof of deadlock freedom for the Linear Convection example in Rocq (about 1000 lines, even with many admitted proof goals), which exceeded that of the entire program correctness proof in DafnyMPI (694 lines).
While it could be both feasible and useful to implement reasoning similar to DafnyMPI with other verification frameworks, we leave this task to future work.

\paragraph{Floating-point Arithmetic}

Throughout our codebase, we use Dafny's \code{real} type to model floating-point numbers.
While common in verification, this approach does not provide a precise model of IEEE-754 or any other floating-point standard because it does not account for floating-point underflow, and does not prevent division-by-zero errors that may result from it.
Several studies exist on automating the verification of programs that use floating-point arithmetic~\cite{floating,floating2}, and it may be a possible direction for future work on DafnyMPI.

\paragraph{Running Time}

Because we model most array operations in Dafny, the runtime performance suffers compared to what could be achieved with native libraries such as Python's Numpy.
This is, in part, by design: we chose to minimize the trust base by only relying on Numpy for single-element set and get operations.
In principle, it should be possible to replace our implementations of such operations as element-wise addition, slicing, and rolling, with corresponding Numpy operations at runtime, which should result in a significant performance boost. 

\section{Related work}

There are several threads of related work.

\paragraph{Dynamic Analysis and Model Checking for MPI}

Researchers have developed a number of verification systems for MPI.
Practical tools such as ISP~\cite{isp}, its successor DAMPI~\cite{dampi}, and MOPPER~\cite{mopper} execute MPI programs under controlled schedulers or otherwise analyze the results of a past execution.
Others, including TASS~\cite{tass}, MPI-SV~\cite{mpi-sv}, and CIVL-C~\cite{civl}, combine symbolic execution with model checking to explore the space of possible message interleavings.
These approaches do not require the same amount of
proof engineering as DafnyMPI, but they are subject to potential state
explosion, even though techniques exist that allow pruning the search space~\cite{verifympi2}.

\paragraph{Lock Orders}

Ordering on locks has been used previously to establish deadlock freedom in several verification contexts. 
A system that enforces a global order on lock acquisitions and receives has been implemented~\cite{leino} for the Chalice verifier and an extension that supports condition variables~\cite{hamin} has been implemented for VeriFast.
Recent work~\cite{horder} uses more complex ordering topologies for managing higher-order locks.
These studies focus on locks in the context of shared memory concurrency and dynamic thread creation, where communication obligations can be established or transferred when a new thread is forked.
In contrast, we target the MPI setting, where the number of processes is arbitrary but fixed throughout execution, meaning the programmer can describe the entire communication topology as a function of process ID without the need to explicitly manage communication for each process. 
Moreover, data transfer in MPI can occur asynchronously relative to the initiating process, a behavior explicitly captured by our operational semantics via the message payload buffer and the separation between the initialization and completion stages of each point-to-point operation.
Finally, DafnyMPI allows simultaneous sends and receives and provides an easy way to integrate collective communication operations.

\paragraph{Concurrency with Dafny}

Several projects have explored concurrency in the context of Dafny.

The DafnyInfoFlow project~\cite{smith-dafny-concurrency} has added support for concurrent reasoning, also using rely-guarantee, but its primary goal is to prevent leakage of information from private variables.
In contrast, DafnyMPI focuses on deadlock freedom, termination, and functional equivalence.

Armada~\cite{armada} is a tool and a programming language whose operational semantics is modeled in Dafny and which allows verification of C-like concurrent code. 
Because Armada is not a Dafny library but a distinct language, Armada programs cannot directly reuse Dafny constructs the way DafnyMPI can.
Moreover, Armada explicitly targets low-level, shared-memory concurrency, which is a fundamentally different setting from MPI.

\paragraph{Other approaches}

Concurrent separation logic (CSL), as exemplified by
Iris~\cite{iris} and Steelcore~\cite{steelcore}, extends Hoare logic with
shared-memory semantics such as resources (heap), resource ownership, and the
separating implication for heap predicates~\cite{ohearn-csl}. 
While CSL models the heap directly, DafnyMPI abstracts away these details, reducing the proof engineering burden but limiting applicability to MPI programs.
CSL also lacks support for liveness or deadlock-freedom, typically requiring a supplemental logic, as in LinearActris~\cite{jacobs-linearactris}, to establish these properties.

Several researchers have developed type systems that aim to prevent common
concurrency bugs such as data races and deadlock. Examples include
TyPiCal~\cite{deadlock-free-session-types}, based on session types;
CLL~\cite{cll}, based on dependent types; and Asynchronous Liquid Separation
Types~\cite{async-liquid}, based on refinement types. In contrast to DafnyMPI,
these approaches do not reason about program equivalence.

Finally, researchers have developed a variety of other formal verification
frameworks for concurrent programs, such as VCC\cite{vcc},
VerCors\cite{vercors}. A major difference between DafnyMPI
and these frameworks lies in their core design. VCC requires annotating C code with
invariants and pre- or postconditions; Vercors verifies Java code with
separation logic annotations.

\section{Conclusion}

We have introduced DafnyMPI, a novel verification library that brings formal reasoning about message-passing interface (MPI) parallel programs into Dafny, a language designed for sequential code.
Our approach enables verification of core MPI properties, including deadlock freedom and functional equivalence with sequential specifications, all without requiring custom concurrency logics.
By leveraging Dafny's built-in verification capabilities and layering MPI-specific invariants on top, we believe our approach achieves a balance of rigor, accessibility, and scalability.

Our formal model demonstrates that a small set of preconditions on message tags, barriers, and buffer usage suffices to guarantee deadlock freedom.
These preconditions are enforced through a core calculus and a set of Dafny assertions.
Experimental results on PDE solvers demonstrate the viability of our method: although proof engineering remains nontrivial, the verified parallel implementations exhibit significant speedups over their sequential counterparts.

Looking forward, we envision extensions of DafnyMPI to support more expressive MPI features (e.g., communicators, wildcards), integrate more seamlessly with scientific programming ecosystems, and further reduce proof-engineering effort.
We hope that DafnyMPI will serve as a foundation for future work on verified high-performance computing and encourage the adoption of formal methods in domains where correctness is critical but traditional techniques are impractical.

\section*{Data-Availability Statement}

DafnyMPI and associated benchmarks are available as a Zenodo artifact~\cite{artifact}.

\begin{acks}
We thank the anonymous reviewers for their valuable feedback. This research was partially funded by the U.S. National Science Foundation under Award Nos. 2313998 and 2513872.
\end{acks}

\bibliographystyle{ACM-Reference-Format}
\bibliography{ref}

\end{document}